\documentclass[letterpaper]{article} % DO NOT CHANGE THIS
\usepackage{aaai25}  % DO NOT CHANGE THIS
\usepackage{times}  % DO NOT CHANGE THIS
\usepackage{helvet}  % DO NOT CHANGE THIS
\usepackage{courier}  % DO NOT CHANGE THIS
\usepackage[hyphens]{url}  % DO NOT CHANGE THIS
\usepackage{graphicx} % DO NOT CHANGE THIS
\usepackage{natbib}  % DO NOT CHANGE THIS AND DO NOT ADD ANY OPTIONS TO IT
\usepackage{caption} % DO NOT CHANGE THIS AND DO NOT ADD ANY OPTIONS TO IT
\usepackage{algorithm}
\usepackage{algorithmic}
\usepackage{bm}
\usepackage{amsmath,amsfonts,amssymb,amsthm}
\usepackage{subcaption}
\usepackage{booktabs}
\usepackage{multirow}
\usepackage{multicol}
\usepackage{lipsum}
\usepackage[utf8]{inputenc}
\usepackage{kotex}
\usepackage{siunitx}
\usepackage{rotating}
\usepackage{mathabx}
\usepackage{pifont}
\usepackage{cleveref}

\crefname{figure}{Fig.}{Figs.} 
\Crefname{figure}{Fig.}{Figs.} 

\crefname{table}{Tab.}{Tabs.}  
\Crefname{table}{Tab.}{Tabs.} 

\crefname{equation}{Eq.}{Eqs.} 
\Crefname{equation}{Eq.}{Eqs.} 
\usepackage[textsize=small]{todonotes}
\usepackage{enumitem}

\usepackage[toc,page]{appendix}

\usepackage{newfloat}
\usepackage{listings}
\DeclareCaptionStyle{ruled}{labelfont=normalfont,labelsep=colon,strut=off} % DO NOT CHANGE THIS
\floatstyle{ruled}
\newfloat{listing}{tb}{lst}{}
\floatname{listing}{Listing}
\title{Future Sight and Tough Fights: Revolutionizing Sequential Recommendation with FENRec}
\author{
    Yu-Hsuan Huang\textsuperscript{\rm 1},
    Ling Lo\textsuperscript{\rm 1},
    Hongxia Xie\textsuperscript{\rm 2},
    Hong-Han Shuai\textsuperscript{\rm 1},
    Wen-Huang Cheng\textsuperscript{\rm 3}
}
\affiliations{
    \textsuperscript{\rm 1}National Yang Ming Chiao Tung University,
    \textsuperscript{\rm 2}Jilin University,
    \textsuperscript{\rm 3}National Taiwan University
    \{shan.ee08, linglo.ee08\}@nycu.edu.tw, hongxiaxie@jlu.edu.cn, hhshuai@nycu.edu.tw, wenhuang@csie.ntu.edu.tw
}

\usepackage{bibentry}
\begin{document}

\maketitle

\begin{abstract}

Sequential recommendation (SR) systems predict user preferences by analyzing time-ordered interaction sequences. A common challenge for SR is data sparsity, as users typically interact with only a limited number of items. 
While contrastive learning has been employed in previous approaches to address the challenges, these methods often adopt binary labels, missing finer patterns and overlooking detailed information in subsequent behaviors of users. Additionally, they rely on random sampling to select negatives in contrastive learning, which may not yield sufficiently hard negatives during later training stages.
In this paper, we propose Future data utilization with Enduring Negatives for contrastive learning in sequential Recommendation (FENRec). Our approach aims to leverage future data with time-dependent soft labels and generate enduring hard negatives from existing data, thereby enhancing the effectiveness in tackling data sparsity. Experiment results demonstrate our state-of-the-art performance across four benchmark datasets, with an average improvement of 6.16\% across all metrics. 
\end{abstract}

\begin{links}
    \link{Code}{\url{https://github.com/uikdwnd/FENRec}}
\end{links}

\section{Introduction}

Recommendation systems enhance user experience by providing personalized suggestions tailored to individual preferences. Given their wide applications in online shopping, streaming services, and social media, extensive research has focused on optimizing their performance and effectiveness~\cite{lu2015recommender, sharma2024survey, alamdari2020systematic}. Among them, Sequential Recommendation (SR) stands out for its ability to capture and utilize the temporal dynamics of user behavior~\cite{yu2019multi, xie2022cl4srec, chen2022iclrec, qiu2022duorec}.

Despite the strength of SR, they still face the common challenge of recommendation systems, the limited nature of user-item interaction. Usually, users only interact with a small subset of all available items. The sparsity in interaction data can impede the training of SR, making them struggle to learn meaningful %patterns and 
preferences, which can affect the accuracy and the relevance of the generated recommendations. 
\begin{figure}[htbp]  
    \centering  
    \includegraphics[width=0.99\columnwidth]{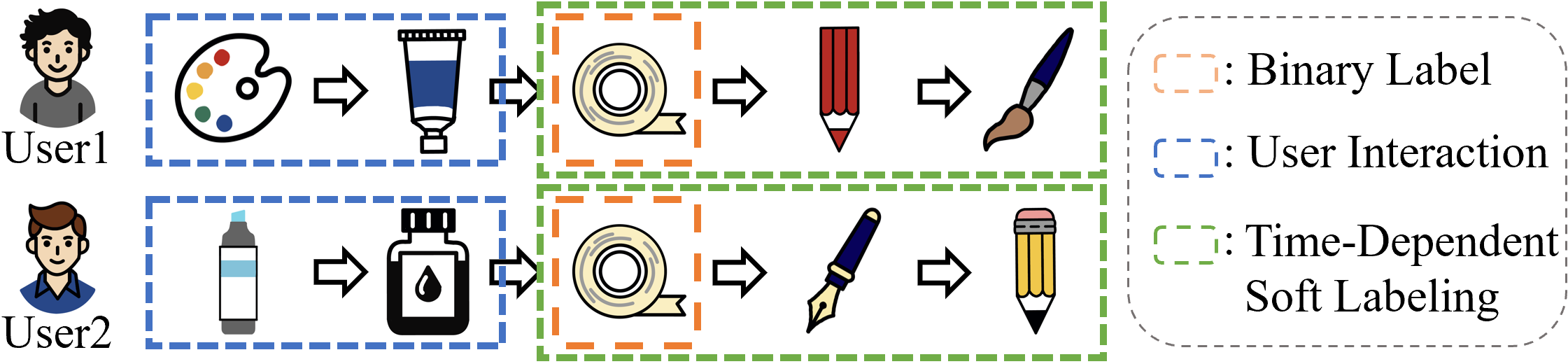}  
    \caption{An illustration of binary labels compared to the Time-Dependent Soft Labeling we propose.}  
    \label{fig:future_interaction}  
\end{figure}
To address the data sparsity issue, previous methods have proposed incorporating contrastive learning, which leverages self-supervised signals to improve model performance~\cite{ xie2022cl4srec, chen2022iclrec, qiu2022duorec}. These methods create positive pairs through augmentation strategies and maximize their agreement. For instance, CL4SRec~\cite{xie2022cl4srec} uses data augmentations such as masking and cropping to generate positive pairs. DuoRec~\cite{qiu2022duorec} utilizes model-level augmentation and regards user sequence with the same next item as augmented positive view.

While contrastive learning has proven effective, we contend that sparse data can be used more efficiently within these methods due to two reasons. First, they use only the last item in sequence as a binary
label and overlook the upcoming interactions. Second, they adopt random sampling for selecting negatives in contrastive learning, which may not yield sufficiently challenging samples later in training, reducing their effectiveness in addressing data sparsity issue. Therefore, to better utilize sparse data, we propose \textbf{F}uture data utilization with \textbf{E}nduring \textbf{N}egatives for contrastive learning in sequential \textbf{Rec}ommendation (FENRec) including two components: Time-Dependent Soft Labeling and Enduring Hard Negatives Incorporation. For the first issue, we propose Time-Dependent Soft Labeling akin to using future interactions as labels, even though these interactions are rooted in the past.
Our method traces back previous subsequences to generate labels based on interactions resembling potential future events, capturing finer patterns in user behavior. Unlike previous methods that assign identical binary labels to the two user sequences in \cref{fig:future_interaction}, our approach incorporates subsequent events, allowing the model to capture more detailed insight in user behavior.
Additionally, to address the second issue, we introduce
Enduring Hard Negatives Incorporation. We generate hard negatives by mixing anchors with negatives throughout training, ensuring they remain consistently challenging.
\begin{figure}[htbp]
    \centering
    \begin{minipage}[b]{0.325\columnwidth}
        \centering
        \includegraphics[width=\textwidth]{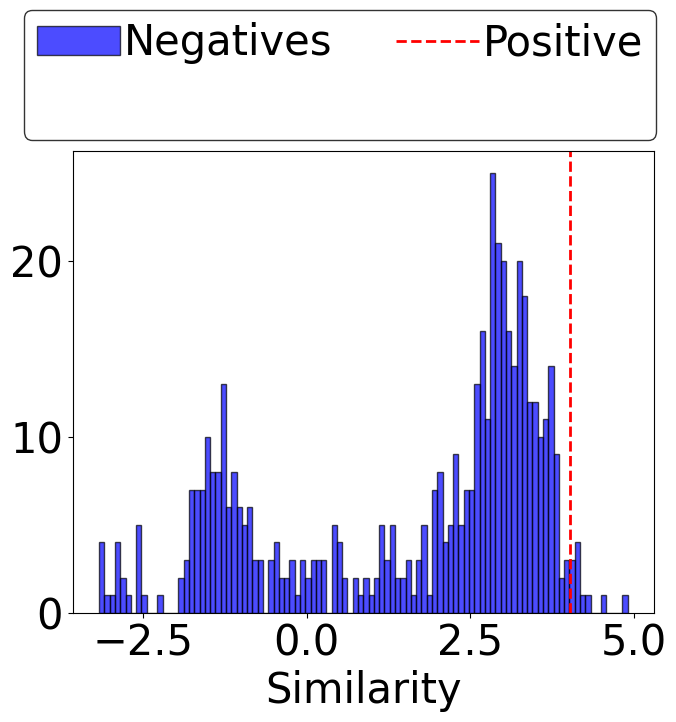} % Replace with your image file
        \subcaption{epoch 5}\label{fig:image1}
    \end{minipage}
    \hfill
    \begin{minipage}[b]{0.325\columnwidth}
        \centering
        \includegraphics[width=\textwidth]{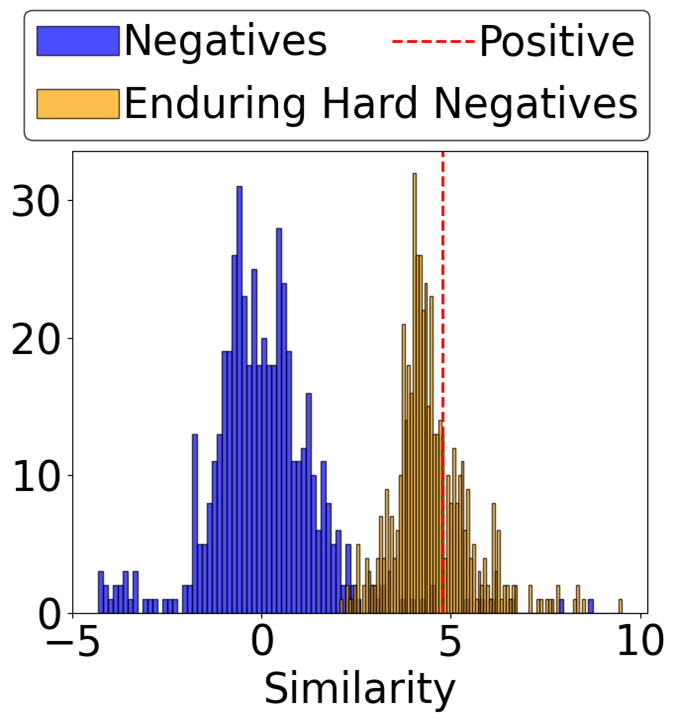} % Replace with your image file
        \subcaption{epoch 25}\label{fig:image2}
    \end{minipage}
    \hfill
    \begin{minipage}[b]{0.325\columnwidth}
        \centering
        \includegraphics[width=\textwidth]{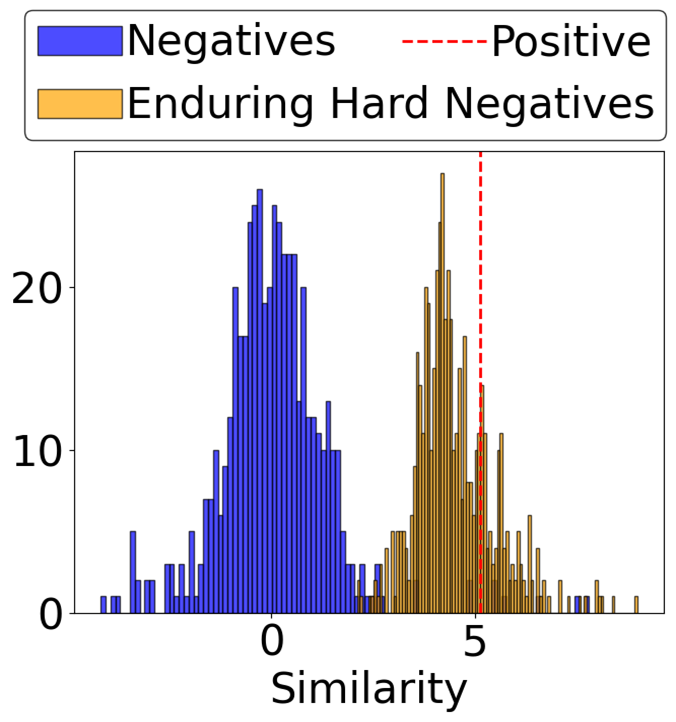} % Replace with your image file
        \subcaption{epoch 50}\label{fig:image3}
    \end{minipage}
    \caption{The distribution of similarity between samples and the anchor sample from a random sample batch at epoch 5, 25, and 50 on the Beauty dataset. Our enduring hard negatives are incorporated following a 20-epoch warm-up period. Similarity is calculated using the inner product. }\label{fig:previous_method_sim}
\end{figure}
\cref{fig:previous_method_sim} shows how the similarity between original negatives and the anchor declines over time, potentially hindering the model's ability to learn discriminative features. In contrast, our hard negatives maintain higher similarity with the anchor, enhancing the model's ability to differentiate user preferences.

The contributions of this work are summarized below:
\begin{itemize}

\item We propose Time-Dependent Soft Labeling to leverage the entire user interaction sequence, generating labels based on interactions that resemble future events.

\item We introduce Enduring Hard Negatives Incorporation to enhance contrastive learning. By generating enduring hard negatives, our method improves the ability of the model to differentiate complex samples.

\item Extensive experiments on four benchmark datasets demonstrate the state-of-the-art performance of our FENRec model, which comprises Time-Dependent Soft Labeling and Enduring Hard Negatives Incorporation.

\end{itemize}

\section{Related Work}

\subsection{Sequential Recommendation}

Sequential recommendation (SR) predicts future user interactions based on past sequences. Recent SR models have improved performance but still face data sparsity issues due to limited user interactions~\cite{tang2018caser, hidasi2015gru4rec, kang2018sasrec, yue2024lrurec, shin2024bsarec}. To address this, some methods use contrastive learning techniques~\cite{xie2022cl4srec, du2022CBiT, liu2021coserec, chen2022iclrec, qiu2022duorec, qin2024icsrec}, introducing self-supervised signals to more effectively utilize the available data for better learning. These approaches use augmentation techniques to generate positive pairs and maximize their agreement. CL4Rec~\cite{xie2022cl4srec} uses data augmentations like item masking to produce positive pairs, ICLRec~\cite{chen2022iclrec} employs clustering to align intent and user representations, DuoRec~\cite{qiu2022duorec} leverages both supervised and unsupervised contrastive learning, and ICSRec~\cite{qin2024icsrec} utilizes intent within subsequences for contrastive learning. However, these contrastive learning methods fail to utilize future interactions and adopt random sampling. Therefore, sparse data may not be fully utilized by these methods.

\subsection{Soft Label for Recommendation Systems}

Recommender systems often rely on one-hot labels, which overlook the ambiguity of unobserved feedback and can lead to poor generalization. To address this, some methods use soft labels to capture user preferences more accurately~\cite{cheng2021softrec, zhou2023MVS, wu2023CSRec}. Soft labels help the model to learn finer-grained user preferences. SoftRec~\cite{cheng2021softrec} uses item, user, and model strategies to create soft labels, MVS~\cite{zhou2023MVS} generates smoothed contexts with a complementary model, and 
CSRec~\cite{wu2023CSRec} employs model-, data-, or training-level teachers to generate confident soft labels for SR. Unlike previous methods, we generate soft labels by leveraging upcoming interactions through a simple yet effective approach that requires no additional modules.

\subsection{Hard Negative Mining in Deep Metrics Learning}

Hard negatives have proved to be helpful in various domains~\cite{suh2019stochastic, xuan2020hard, zhan2021optimizing, robinson2021contrastive}. 
Due to the advantage of using hard negatives, Mochi~\cite{kalantidis2020hard} creates synthetic hard negatives by mixing the hardest negatives. MixCSE~\cite{zhang2022unsupervised} mix positive and negative samples to generate hard negatives. Despite the advantages of hard negatives and the prevalence of contrastive learning in SR, integrating hard negatives into these frameworks remains underexplored.

\section{Preliminaries}
\begin{figure*}[htbp]  
    \centering  
    \includegraphics[width=0.8\textwidth]{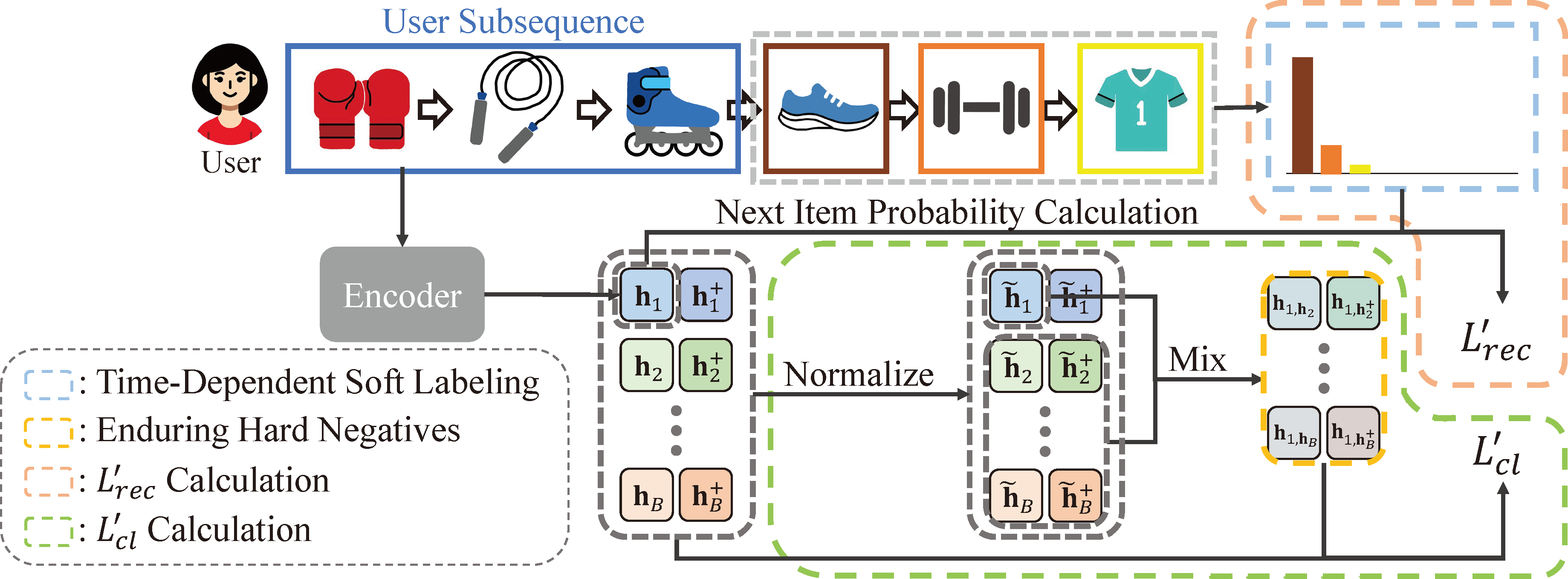}  
    \caption{The framework of our method, FENRec, the user sequences will first be split into subsequences and encoded into representations, while soft labels are generated based on the subsequences. Next, the enduring hard negatives are produced and incorporated into contrastive learning framework. Finally, $L_{rec}^{\prime}$ and $L_{cl}^{\prime}$ are calculated.}  
    \label{fig:framework}  
\end{figure*}

\subsection{Sequential Recommendation}
Sequential recommendation (SR) aims to predict the next item a user will interact with by analyzing their historical interaction sequences. 
Denote the set of items as $\mathcal{V}=\left\{v_1, v_2 \dots, v_N\right\}$ and the set of users as $\mathcal{U}$. The past interactions of each user $u \in \mathcal{U}$ can be arranged chronologically, forming a user interaction sequence $\mathcal{S}^{u}  = [v^{u}_1, v^{u}_2, \dots, v^{u}_{|\mathcal{S}^{u}|}]$, where $v^{u}_i$ represents the $i$-th interacted item in the sequence and $|\mathcal{S}^{u}|$ is the length of $\mathcal{S}^{u}$. Given a user interaction sequence $\mathcal{S}^{u}$, the goal of SR is to predict the item the user $u$ is most likely to interact with at time step $|\mathcal{S}^{u}| + 1$. It can be formulated as follows:
\begin{align}
    \mathop{\arg\max}_{v \in \mathcal{V}}\mathbb{P}(v^{u}_{|\mathcal{S}^{u}|+1}=v \mid \mathcal{S}^{u}).
\label{eq:problem_formulation}
\end{align}

To effectively train models to achieve this goal, SR methods often divide the original sequence into multiple subsequences to enrich training data and train the model using them. Formally, given the interaction sequence $\mathcal{S}^u$ of user $u$, the subsequences can be constructed as follows:
\begin{align}
\mathcal{DS}(\mathcal{S}^u) = \left\{ [v^{u}_1, v^{u}_2, \dots, v^{u}_t] \mid 1 \leq t \leq |\mathcal{S}^{u}| \right\}.
\end{align}
During the training phase, the next item of each subsequence is regarded as the target item for prediction. 

After dividing user sequences, SR model is trained by calculating the probability of each item being the next of the subsequence and computing the cross-entropy loss based on this probability. Specifically, current SR systems first encode the interaction sequences into user representations using deep neural networks (e.g., RNN or transformer) and calculate the probability of the next item based on it~\cite{hidasi2015gru4rec, kang2018sasrec}.
Formally, a sequence encoder $f_{\theta}(\cdot)$ first encodes the interaction sequence $\mathcal{S}^{us} \in \mathcal{DS}(\mathcal{S}^u)$ and generates the user representation $\mathbf{h}_{us}$, denoted as $\mathbf{h}_{us} = f_{\theta}(\mathcal{S}^{us})$. 
Then, the probability that the next item of sequence $\mathcal{S}^{us}$ is item $v_j$ is given by:
\begin{align}
    \hat{\mathbf{y}}^{us}_j = softmax(\mathbf{h}_{us}^T \mathbf{v}_{j}),
\end{align}
where $\mathbf{v}_{j}$ denotes the item embedding of the item $v_j$, and the cross-entropy loss is calculated as follows:
\begin{align}
    L_{rec} = -\sum_{j=1}^{N} \mathbf{y}^{us}_j \log \hat{\mathbf{y}}^{us}_j
\label{eq:next_item_ce}.
\end{align}
$N$ represents the total number of items. $\mathbf{y}^{us}_j$ denotes the true label of the interaction sequence $\mathcal{S}^{us}$ for $v_j$.
If $v_i$ is the target item of $\mathcal{S}^{us}$ (i.e., $v^{u}_{|\mathcal{S}^{us}|+1}$), then $\mathbf{y}^{us}_i = 1$ and $\mathbf{y}^{us}_j = 0$ for all other items (i.e., $j \neq i$). Minimizing the cross-entropy loss is equivalent to optimizing SR system with \cref{eq:problem_formulation}.

\subsection{Contrastive Learning in SR}
One challenge in SR is data sparsity. Typically, users interact with only a small fraction of available items, leading to sparse history data. 
Since the sparsity can significantly hinder the learning of deep neural networks, recent works adopt contrastive learning to improve the learned features~\cite{ chen2022iclrec, qiu2022duorec}.
Generally, given an anchor sample\footnote{The anchor sample can be either the original representation or an augmented version from techniques like masking or cropping.} $\mathbf{h}_{i}$ and the positive view of the anchor $\mathbf{h}_{i}^{+}$ generated by augmentation strategies\footnote{Augmentation strategies such as model-level (e.g. dropout) or data-level augmentation (e.g. masking or cropping).}, other samples within the batch are considered negative views, and the contrastive loss can be formulated as follows:
\begin{align}
    L_{cl}(\mathbf{h}_{i}, \mathbf{h}_{i}^{+}) = -\log \frac{e^{(\mathbf{h}_{i}^T \mathbf{h}_{i}^{+}) / \tau_1}}{e^{(\mathbf{h}_{i}^T \mathbf{h}_{i}^{+}) / \tau_1} + \sum_{\mathbf{n} \in \mathbf{H}_i} e^{(\mathbf{h}_{i}^T \mathbf{n}) / \tau_1}},
\label{eq:infonce}
\end{align}
where $\mathbf{H}_i = \left\{ \mathbf{h}_j, \mathbf{h}_j^+ \mid j \neq i, 0 < j \leq B \right\}$, and $B$ denotes the batch size.  $\tau_1$ is the temperature parameter.

\section{Method}

\cref{fig:framework} shows the overall framework of our method, \textbf{F}uture data utilization with \textbf{E}nduring \textbf{N}egatives for contrastive learning in sequential \textbf{Rec}ommendation, FENRec. 
To enhance contrastive learning in SR system, we adopt a comprehensive strategy. 

First, we divide user interaction sequences into subsequences and encode them into user representations. To leverage future interactions, we generate time-dependent soft labels that assign probabilities to not only the next item but also the future interactions, as depicted in the top half of \cref{fig:framework}. These soft labels are then employed to calculate the revised cross-entropy loss $L_{rec}^{\prime}$, allowing the model to consider finer patterns in users' upcoming behaviors. 

Next, we use user representations to generate enduring hard negatives, as shown in the bottom half of \cref{fig:framework}. This process consistently produces enduring hard negatives that are more similar to the anchor than the original negatives, ensuring that the model continues to encounter challenging samples throughout the training process, thereby improving its ability to differentiate user preferences. After generating these enduring hard negatives, the revised contrastive loss $L_{cl}^{\prime}$ is then calculated using them. Finally, the total loss is computed by considering both $L_{rec}^{\prime}$ and $L_{cl}^{\prime}$.

\subsection{Future Interaction Utilization}

By tracing back previous subsequences, the future interactions of these subsequences are identified. However, most existing SR models often overlook this valuable information. To fully leverage interaction data, we propose generating soft labels based on forthcoming interactions. 

\subsubsection{Time-Dependent Soft Labeling.}
Previous methods fail to leverage forthcoming interactions effectively since they use binary labels and focus solely on predicting the immediate next item. Therefore, to better utilize future data, we generate time-dependent soft labels by considering items the user will interact with in the near future (i.e., the second and third subsequent items) as items of interest. Items closer in time to the next item are assigned with higher probabilities. 

We define the time-dependent soft label as:

\begin{align}
{\mathbf{y}_i^{us}}^{\prime} =
\begin{cases}
p(i) & \text{if } v_i \in \{v^{u}_{|\mathcal{S}^{us}| + 1}, \ldots, \\ &\phantom{\text{if } v_i \in \{} v^{u}_{min(|\mathcal{S}^{us}| + 3, |\mathcal{S}^{u}| + 1)}\} \\
0 & \text{otherwise}
\end{cases},
\end{align}
where
\begin{align}
p(i) = \frac{\gamma^{(pos(v_i) - (|\mathcal{S}^{us}| + 1))}}{\sum_{j=|\mathcal{S}^{us}| + 1}^{min(|\mathcal{S}^{us}| + 3, |\mathcal{S}^{u}| + 1)} \gamma^{(j - (|\mathcal{S}^{us}| + 1))}}. 
\end{align}
$pos(v_i)$ is the position of $v_i$ in the user sequence $\mathcal{S}^{u}$ (i.e., $pos(v_i)=|\mathcal{S}^{us}| + 1$, if $v^u_{|\mathcal{S}^{us}| + 1} = v_i$). $\mathcal{S}^{us} \in \mathcal{DS}(\mathcal{S}^u)$ %\cup \mathcal{S}^{u} 
represents the subsequence derived from the original user sequence, and $|\mathcal{S}^{us}|$ is the length of the sequence $\mathcal{S}^{us}$. %$min(.,.)$ represents the minimum of the two given values. 
$\gamma$ is a smoothing hyperparameter that controls the distribution of time-dependent soft labels. A higher value of $\gamma$ results in a smoother probability distribution of soft labels, indicating greater uncertainty. Conversely, a lower value of $\gamma$ leads to a more concentrated distribution, suggesting more confident predictions. We use time-dependent soft labels to replace the original binary labels. Hence, \cref{eq:next_item_ce} can be altered as:
\begin{align}
    L_{rec}^{\prime} = -\sum_{i=1}^{N} {\mathbf{y}_i^{us}}^{\prime} \log \hat{\mathbf{y}}_i^{us}.
\label{eq:soft_next_item_ce}
\end{align}

\subsection{Enduring Hard Negatives Incorporation}

Previous contrastive learning methods in SR often fail to provide challenging negatives in later training stages, limiting their effectiveness in alleviating data sparsity. Therefore, to better address data sparsity and enhance the model's ability to learn discriminative features, we propose incorporating enduring hard negatives during training. Specifically, We generate enduring hard negatives from existing data, and assign higher weights to them during training, ensuring that the model consistently learns from the challenging samples, thereby optimizing the use of sparse data.

\subsubsection{Enduring Hard Negatives.}

For contrastive learning in SR, we define hard negatives as user representations similar to the anchor sample but not its augmented view. Inspired by previous works~\cite{kalantidis2020hard, zhang2022unsupervised}, we proposed generating enduring hard negatives that remain challenging throughout training. As shown in the lower half of \cref{fig:framework}, we first normalize the user representations in the batch and mix the anchor with the negative samples to produce enduring hard negatives $\mathbf{h}^{-}_{i,\mathbf{n}}$. It can be formulated as follows:

\begin{align}
    \mathbf{\tilde{h}}_{u} = \frac{\mathbf{h}_{u}}{\|\mathbf{h}_{u}\|_2},
\end{align}

\begin{align}
\mathbf{h}^{-}_{i,\mathbf{n}} = \frac{\lambda \mathbf{\tilde{h}}_{i} + (1 - \lambda) \mathbf{\tilde{n}}}{\|\lambda \mathbf{\tilde{h}}_{i} + (1 - \lambda) \mathbf{\tilde{n}}\|_2} \cdot \|\mathbf{n}\|_2,
\end{align}
where $\mathbf{n} \in \mathbf{H}_i$.  $\lambda$ is a hyperparameter between 0 and 1 that controls the proportion of the anchor sample in the generated enduring hard negatives. We first normalize representations to ensure that anchor and negative samples contribute equally to the generated enduring hard negatives. Then, by mixing the anchor and negative samples, the enduring hard negatives maintain a consistently higher similarity to the anchor sample than the original negatives\footnote{For more details, interested readers can refer to Appendix.}. Finally, we divide the mixture of anchor and negative samples (i.e., $\lambda \mathbf{\tilde{h}}_{i} + (1 - \lambda) \mathbf{\tilde{n}}$) by their norm and multiply it by the norm of the negative samples, allowing enduring hard negatives to maintain their original norm range, ensuring their data distribution resembles that of the true user representations.

Therefore, the contrastive loss can be formulated as:
\begin{align}
L_{cl}(\mathbf{h}_{i}, \mathbf{h}_{i}^{+}) &= -\log \frac{e^{ \left(\mathbf{h}_{i}^T \mathbf{h}_{i}^{+}\right)/\tau_1 } }{C + \mu  \sum_{\mathbf{n} \in \mathbf{H}_i} e^{ \left(\mathbf{h}_{i}^T  \text{SG}(\mathbf{h}_{i, \mathbf{n}}^{-})\right)/\tau_1 }} ,
\label{eq:infonce_with_mix_neg}
\end{align}
where
\begin{align}
C &= e^{ \left(\mathbf{h}_{i}^T \mathbf{h}_{i}^{+}\right)/\tau_1 }  + \sum_{\mathbf{n} \in \mathbf{H}_i} e^{\left( \mathbf{h}_{i}^T \mathbf{n}\right)/\tau_1 }.
\label{eq:C_eq}
\end{align}
$\mu$ is a hyperparameter to control the portion of enduring hard negatives in the contrastive loss. The ``stop gradient,"‵ denoted by $\text{SG}(.)$, stops the gradient from propagating through the generated enduring hard negatives. This avoids incorrect backpropagation signals. Without $\text{SG}(.)$, the anchor sample would be affected by the gradients passed through the generated enduring hard negatives. Since these negatives partially contain the anchor, the anchor would be incorrectly pushed away from itself. 
\begin{table*}[t]
    \centering
    \footnotesize
    \resizebox{\textwidth}{!}{%
    % \begin{tabular}{c|l|ccccccc|c|c}
    \begin{tabular}{ll| ccccc|cccccccc|c|cc}
    \specialrule{1pt}{1pt}{2pt}
    % \multicolumn{2}{c|}{Dataset} & \multicolumn{6}{c|}{\texttt{Electricity}}&\multicolumn{6}{c}{\texttt{Energy}} \\
    Dataset & Metric & GRU4Rec & Caser & LRURec & SASRec & BSARec & BERT4Rec & MAERec & CBiT & CL4SRec & CoSeRec & ICLRec & DuoRec & ICSRec & MVS & FENRec & Improv.    \\
    \specialrule{1pt}{1pt}{2pt}

\multirow{6}{*}{Sports} 
 & HIT@5 & 0.0116 & 0.0123 & 0.0389 & 0.0189 & \underline{0.0400} & 0.0264 & 0.0285 & 0.0235 & 0.0235 & 0.0264 & 0.0271 & 0.0311 & 0.0388 & 0.0384 & \textbf{0.0431} & 7.75\% \\
 & HIT@10 & 0.0197 & 0.0210 & 0.0551 & 0.0307 & \underline{0.0583} & 0.0408 & 0.0435 & 0.0365 & 0.0375 & 0.0403 & 0.0422 & 0.0446 & 0.0551 & 0.0548 & \textbf{0.0621} & 6.52\% \\
 & HIT@20 & 0.0320 & 0.0336 & 0.0771 & 0.0491 & \underline{0.0830} & 0.0622 & 0.0645 & 0.0528 & 0.0575 & 0.0605 & 0.0632 & 0.0640 & 0.0767 & 0.0775 & \textbf{0.0890} & 7.23\% \\
 & NDCG@5 & 0.0074 & 0.0078 & 0.0276 & 0.0122 & \underline{0.0280} & 0.0175 & 0.0191 & 0.0157 & 0.0156 & 0.0177 & 0.0179 & 0.0220 & 0.0272 & 0.0268 & \textbf{0.0299} & 6.79\% \\
 & NDCG@10  & 0.0100 & 0.0105 & 0.0329 & 0.0161 & \underline{0.0339} & 0.0221 & 0.0239 & 0.0198 & 0.0201 & 0.0221 & 0.0227 & 0.0263 & 0.0324 & 0.0321 & \textbf{0.0361} & 6.49\% \\
 & NDCG@20  & 0.0131 & 0.0137 & 0.0384 & 0.0207 & \underline{0.0401} & 0.0275 & 0.0292 & 0.0239 & 0.0251 & 0.0272 & 0.0280 & 0.0312 & 0.0379 & 0.0378 & \textbf{0.0429} & 6.98\% \\

\specialrule{0.5pt}{1pt}{2pt} 

\multirow{6}{*}{Beauty} 
 & HIT@5 & 0.0188 & 0.0234 & 0.0671 & 0.0359 & \underline{0.0707} & 0.0489 & 0.0557 & 0.0612 & 0.0492 & 0.0459 & 0.0493 & 0.0560 & 0.0681 & 0.0691 & \textbf{0.0728} & 2.97\% \\
 & HIT@10 & 0.0315 & 0.0386 & 0.0928 & 0.0580 & \underline{0.0978} & 0.0735 & 0.0789 & 0.0871 & 0.0706 & 0.0696 & 0.0726 & 0.0800 & 0.0936 & 0.0961 & \textbf{0.1019} & 4.19\% \\
 & HIT@20 & 0.0516 & 0.0585 & 0.1257 & 0.0905 & \underline{0.1345} & 0.1065 & 0.1094 & 0.1202 & 0.0990 & 0.1020 & 0.1055 & 0.1088 & 0.1273 & 0.1305 & \textbf{0.1393} & 3.57\% \\
 & NDCG@5 & 0.0114 & 0.0148 & 0.0481 & 0.0233 & \underline{0.0503} & 0.0330 & 0.0397 & 0.0435 & 0.0348 & 0.0301 & 0.0325 & 0.0406 & 0.0487 & 0.0494 & \textbf{0.0514} & 2.19\% \\
 & NDCG@10  & 0.0154 & 0.0197 & 0.0564 & 0.0304 & \underline{0.0590} & 0.0409 & 0.0472 & 0.0518 & 0.0417 & 0.0378 & 0.0400 & 0.0483 & 0.0569 & 0.0581 & \textbf{0.0608} & 3.05\% \\
 & NDCG@20  & 0.0205 & 0.0248 & 0.0647 & 0.0385 & \underline{0.0682} & 0.0492 & 0.0548 & 0.0602 & 0.0488 & 0.0460 & 0.0483 & 0.0555 & 0.0654 & 0.0667 & \textbf{0.0702} & 2.93\% \\
\specialrule{0.5pt}{1pt}{2pt} 

\multirow{6}{*}{Toys} 
 & HIT@5 & 0.0164 & 0.0180 & 0.0707 & 0.0481 & \underline{0.0792} & 0.0476 & 0.0589 & 0.0632 & 0.0630 & 0.0576 & 0.0576 & 0.0609 & 0.0776 & 0.0748 & \textbf{0.0818} & 3.28\% \\
 & HIT@10 & 0.0277 & 0.0277 & 0.0941 & 0.0699 & \underline{0.1066} & 0.0690 & 0.0823 & 0.0865 & 0.0863 & 0.0818 & 0.0826 & 0.0816 & 0.1035 & 0.1008 & \textbf{0.1109} & 4.03\% \\
 & HIT@20 & 0.0461 & 0.0421 & 0.1228 & 0.0982 & \underline{0.1405} & 0.0974 & 0.1108 & 0.1166 & 0.1143 & 0.1121 & 0.1137 & 0.1080 & 0.1355 & 0.1323 & \textbf{0.1462} & 4.06\% \\
 & NDCG@5 & 0.0104 & 0.0117 & 0.0523 & 0.0326 & \underline{0.0574} & 0.0332 & 0.0424 & 0.0453 & 0.0447 & 0.0399 & 0.0393 & 0.0449 & 0.0566 & 0.0547 & \textbf{0.0592} & 3.14\% \\
 & NDCG@10  & 0.0140 & 0.0149 & 0.0598 & 0.0396 & \underline{0.0662} & 0.0401 & 0.0499 & 0.0529 & 0.0522 & 0.0477 & 0.0473 & 0.0515 & 0.0650 & 0.0631 & \textbf{0.0686} & 3.63\% \\
 & NDCG@20  & 0.0187 & 0.0185 & 0.0671 & 0.0468 & \underline{0.0747} & 0.0472 & 0.0570 & 0.0605 & 0.0592 & 0.0553 & 0.0552 & 0.0582 & 0.0731 & 0.0710  & \textbf{0.0775} & 3.75\% \\

\specialrule{0.5pt}{1pt}{2pt} 

\multirow{6}{*}{Yelp} 
 & HIT@5 & 0.0129 & 0.0137 & 0.0240 & 0.0147 & 0.0252 & 0.0215 & 0.0255 & 0.0164 & 0.0238 & 0.0221 & 0.0232 & 0.0236 & \underline{0.0260} & 0.0243 & \textbf{0.0286} & 10.00\% \\
 & HIT@10 & 0.0227 & 0.0246 & 0.0410 & 0.0254 & \underline{0.0432} & 0.0361 & 0.0423 & 0.0281 & 0.0404 & 0.0375 & 0.0394 & 0.0402 & 0.0431 & 0.0409 & \textbf{0.0485} & 12.27\% \\
 & HIT@20 & 0.0386 & 0.0419 & 0.0652 & 0.0418 & \underline{0.0704} & 0.0608 & 0.0687 & 0.0474 & 0.0655 & 0.0618 & 0.0645 & 0.0663 & 0.0700 & 0.0654 & \textbf{0.0776} & 10.23\% \\
 & NDCG@5 & 0.0082 & 0.0086 & 0.0152 & 0.0091 & 0.0159 & 0.0134 & 0.0162 & 0.0102 & 0.0150 & 0.0141 & 0.0147 & 0.0150 & \underline{0.0165} & 0.0156 & \textbf{0.0182} & 10.30\% \\
 & NDCG@10  & 0.0113 & 0.0120 & 0.0207 & 0.0125 & 0.0217 & 0.0181 & 0.0216 & 0.0140  & 0.0204 & 0.0190 & 0.0198 & 0.0202 & \underline{0.0220} & 0.0210  & \textbf{0.0246} & 11.82\% \\
 & NDCG@20  & 0.0153 & 0.0164 & 0.0267 & 0.0166 & 0.0285 & 0.0243 & 0.0282 & 0.0188 & 0.0266 & 0.0251 & 0.0262 & 0.0268 & \underline{0.0288} & 0.0271 & \textbf{0.0319} & 10.76\% \\

\specialrule{1pt}{1pt}{1pt} 

    \end{tabular}%
    }
    \caption{Performance comparison of different methods on 4 datasets. The best results are in boldface and the second-best results are underlined.`Improv.' indicates the relative improvement against the best baseline. More details are in Appendix.}
    
    \label{tab:main_exp}
\end{table*}
\subsubsection{Hard Negative Upweighting Contrastive Loss.}

In addition to enduring hard negatives, we further refine the contrastive loss by focusing on hard negatives, enabling the model to learn more effectively from challenging samples. Inspired by focal-infoNCE~\cite{hou2023improving}, we propose using hard negative upweighting contrastive loss in the contrastive learning framework for SR. The loss is as follows:

\begin{align}
\begin{split}
L_{h} &= -\log \frac{e^{\mathbf{h}_{i}^T \mathbf{h}_{i}^{+}\cdot \text{s}^{+}_{i}/\tau_1}}{e^{\mathbf{h}_{i}^T \mathbf{h}_{i}^{+} \cdot \text{s}^{+}_{i}/\tau_1} + \sum_{\mathbf{n} \in \mathbf{H}_i} e^{(\mathbf{h}_{i}^T \mathbf{n}\cdot \text{s}(\mathbf{h}_{i}, \mathbf{n}) + m)/\tau_1}} \\
\end{split},
\end{align} 
where $\text{s}(\mathbf{a}, \mathbf{b}) = \tanh(\frac{\mathbf{a}^T \mathbf{b}}{\tau_2})$ and $\text{s}^{+}_{i} = \text{s}(\mathbf{h}_{i}, \mathbf{h}_{i}^{+})$. We use $\tanh$ to constrain similarity values between -1 and 1. The hyperparameter $\tau_2$ scales the similarity measure, adjusting the sensitivity of $\tanh$ to input values. $m$ is a hyperparameter that allows flexible adjustment for the re-weighting approach.
Therefore, we reformulate \cref{eq:infonce_with_mix_neg} and \cref{eq:C_eq} as:

\begin{align}
L_{cl}^{\prime}(\mathbf{h}_{i}, \mathbf{h}_{i}^{+}) &= -\log \frac{e^{ \mathbf{h}_{i}^T \mathbf{h}_{i}^{+}\cdot \text{s}^{+}_{i}/\tau_1 } }{C^{\prime} + \mu \sum_{\mathbf{n} \in \mathbf{H}_i} e^{ \mathbf{h}_{i}^T  \text{SG}(\mathbf{h}_{i, \mathbf{n}}^{-})\cdot \text{s}^-_{i, \mathbf{n}}/\tau_1 }}, 
\label{eq:focalinfonce_with_mix_neg}
\end{align}
where
\begin{align}
C^{\prime} &= e^{\mathbf{h}_{i}^T \mathbf{h}_{i}^{+} \cdot \text{s}^{+}_{i}/\tau_1} + \sum_{\mathbf{n} \in \mathbf{H}_i} e^{(\mathbf{h}_{i}^T \mathbf{n}\cdot \text{s}(\mathbf{h}_{i}, \mathbf{n}) + m)/\tau_1},
\label{eq:focal_C_eq}
\end{align}
and $\text{s}^-_{i, \mathbf{n}}=\text{s}(\mathbf{h}_{i}, \mathbf{h}_{i, \mathbf{n}}^{-})$. Using \cref{eq:focalinfonce_with_mix_neg}, we upweight hard negative samples in contrastive learning, enabling the model to better distinguish between similar representations.

\subsection{Multi-Task Learning}
Similar to previous works~\cite{qin2024icsrec, qiu2022duorec}, we adopt a multi-task learning technique to simultaneously optimize the revised cross-entropy loss and the auxiliary contrastive learning objectives. \cref{eq:soft_next_item_ce} focuses on optimizing the main task of predicting the next item, while \cref{eq:focalinfonce_with_mix_neg} optimizes the auxiliary contrastive learning task. The overall training loss function is formulated as:
\begin{align}
L_{total} = L_{rec}^{\prime} + \alpha\left(L_{cl}^{\prime}(\mathbf{h}_{i}, \mathbf{h}_{i}^{+}) + L_{cl}^{\prime}(\mathbf{h}_{i}^{+}, \mathbf{h}_{i})\right),
\end{align}
where $\alpha$ is a hyperparameter controlling the weight of the revised contrastive loss relative to the revised cross-entropy loss.
Notably, our FENRec can be integrated into different contrastive learning frameworks in SR by replacing the contrastive loss $L_{cl}$ with our proposed $L_{cl}^{\prime}$ and the cross-entropy loss $L_{rec}$ with our proposed $L_{rec}^{\prime}$. 

\section{Experiments}

\subsection{Experimental Setup}
\subsubsection{Metrics.}
We assess model performance using Hit Ratio@K (HR@K) and Normalized Discounted Cumulative Gain@K (NDCG@K), with K selected from \{5, 10, 20\}. Following~\cite{krichene2020sampled, wang2019neural}, we evaluate the ranking of predictions across whole item set.
\subsubsection{Datasets.}
The Amazon dataset is a popular dataset in SR research. In this study, following previous work~\cite{chen2022iclrec, qin2024icsrec}, we select three categories: \textsl{Sports}, \textsl{Beauty}, and \textsl{Toys}. In addition, Yelp is a dataset for business recommendation and we use records after January 1st, 2019 according to previous studies~\cite{chen2022iclrec}.
\subsubsection{Baselines.\footnote{Details of the baselines are provided in Appendix.}}

We compare our method, FENRec, with state-of-the-art SR approaches, broadly divided into 3 categories:

\begin{itemize}
\item \textbf{General sequential models}:       
 Caser~\cite{tang2018caser}, GRURec~\cite{hidasi2015gru4rec}, SASRec~\cite{kang2018sasrec}, LRURec~\cite{yue2024lrurec}, and BSARec~\cite{shin2024bsarec}
\item \textbf{Sequential models with self-supervised learning}: 
 BERT4Rec~\cite{sun2019bert4rec}, MAERec~\cite{ye2023graph}, CL4SRec~\cite{xie2022cl4srec}, CoSeRec~\cite{liu2021coserec}, CBiT~\cite{du2022CBiT}, DuoRec~\cite{qiu2022duorec}, ICLRec~\cite{chen2022iclrec}, and ICSRec~\cite{qin2024icsrec}
 \item \textbf{Sequential models with label smoothness}: 
    MVS~\cite{zhou2023MVS}
\end{itemize}

\subsubsection{Implementation Details.\footnote{More implementation details are provided in Appendix.}}
In our experiments, we configure the embedding dimension to 64 and the maximum user sequence length to 50 across all methods. We adjust the batch size to 256, although for the MVS system, we use a reduced batch size of 64 for the Sports and Yelp datasets during training due to memory constraints. For our evaluation framework, we integrate our FENRec method into the ICSRec system to facilitate a comprehensive comparison. To ensure uniformity in our representations, we employ a noise-based negative sampling method on sentence representations as outlined by~\cite{zhou2022dclr}. Parameter tuning is meticulously carried out; $\tau_2$ is varied within the set \{8, 10\}, while $\tau_1$ was fixed at 1, $\mu$ at 0.1, and $m$ at 0.2. The parameters $\gamma$ and $\lambda$ are each tuned over the range \{0.1, 0.2, 0.3, 0.4, 0.5\}. Following a 20-epoch warm-up period, we incorporate enduring hard negatives into the training process. All experiments were conducted three times to ensure reliability, and results were averaged to provide a fair comparison.

\subsection{Comparison to SOTA}
\cref{tab:main_exp} shows the recommendation performance for all datasets across all metrics. First, contrastive learning methods CL4SRec, CoSeRec, ICLRec, DuoRec, and ICSRec significantly improve upon their backbone, SASRec, by effectively mitigating data sparsity issues in SR, demonstrating the effectiveness of contrastive learning. Second, MVS also demonstrates notable improvements over its backbone, SASRec, due to the context and label smoothness it introduces, which helps the model capture user preferences more effectively and reduces overfitting. Finally, FENRec outperforms all other models across all metrics, achieving an average improvement of 6.34\% in HIT and 5.99\% in NDCG over the second-best results on all datasets. Notably, it shows greater gains in metrics at larger k values compared to those at 5. This improvement is likely due to the inclusion of enduring hard negatives and time-dependent soft labels, which enhance the model's ability to differentiate similar user preferences. The increased discriminative power also improves rankings across the list, not just at the top. 

\subsection{Compatibility to Existing Methods}

\begin{table}[t]
    \centering
    \footnotesize
    \resizebox{\columnwidth}{!}{%
    % \begin{tabular}{c|l|ccccccc|c|c}
    \begin{tabular}{ll| ccc|ccc}
    \specialrule{1pt}{1pt}{2pt}
    % \multicolumn{2}{c|}{Dataset} & \multicolumn{6}{c|}{\texttt{Electricity}}&\multicolumn{6}{c}{\texttt{Energy}} \\
    Dataset & Metric & CL4SRec & + FENRec & Improv. & DuoRec & + FENRec & Improv.    \\
    \specialrule{1pt}{1pt}{2pt}

\multirow{6}{*}{Sports} 
 & HIT@5 & 0.0235 & \textbf{0.0249} & 5.96\% & 0.0311 & \textbf{0.0330} & 6.11\% \\
 & HIT@10 & 0.0375 & \textbf{0.0389} & 3.73\% & 0.0446 & \textbf{0.0473} & 6.05\% \\
 & HIT@20 & 0.0575 & \textbf{0.0585} & 1.74\% & 0.0640 & \textbf{0.0681} & 6.41\% \\
 & NDCG@5 & 0.0156 & \textbf{0.0167} & 7.05\% & 0.0220 & \textbf{0.0231} & 5.00\% \\
 & NDCG@10  & 0.0201 & \textbf{0.0212} & 5.47\% & 0.0263 & \textbf{0.0277} & 5.32\% \\
 & NDCG@20  & 0.0251 & \textbf{0.0261} & 3.98\% & 0.0312 & \textbf{0.0329} & 5.45\% \\
\specialrule{0.5pt}{1pt}{2pt} 
\multirow{6}{*}{Beauty} 
 & HIT@5 & \textbf{0.0492} & \textbf{0.0492} & 0.00\% & 0.0560 & \textbf{0.0586} & 4.64\% \\
 & HIT@10 & 0.0706 & \textbf{0.0718} & 1.70\% & 0.0800 & \textbf{0.0834} & 4.25\% \\
 & HIT@20 & 0.0990 & \textbf{0.1028} & 3.84\% & 0.1088 & \textbf{0.1145} & 5.24\% \\
 & NDCG@5 & \textbf{0.0348} & 0.0345 & -0.86\% & 0.0406 & \textbf{0.0421} & 3.69\% \\
 & NDCG@10  & 0.0417 & \textbf{0.0418} & 0.24\% & 0.0483 & \textbf{0.0501} & 3.73\% \\
 & NDCG@20  & 0.0488 & \textbf{0.0496} & 1.64\% & 0.0555 & \textbf{0.0579} & 4.32\% \\
\specialrule{0.5pt}{1pt}{2pt} 
\multirow{6}{*}{Toys} 
 & HIT@5 & 0.0630 & \textbf{0.0641} & 1.75\% & 0.0609 & \textbf{0.0648} & 6.40\% \\
 & HIT@10 & 0.0863 & \textbf{0.0890} & 3.13\% & 0.0816 & \textbf{0.0871} & 6.74\% \\
 & HIT@20 & 0.1143 & \textbf{0.1208} & 5.69\% & 0.1080 & \textbf{0.1177} & 8.98\% \\
 & NDCG@5 & 0.0447 & \textbf{0.0455} & 1.79\% & 0.0449 & \textbf{0.0477} & 6.24\% \\
 & NDCG@10  & 0.0522 & \textbf{0.0535} & 2.49\% & 0.0515 & \textbf{0.0549} & 6.60\% \\
 & NDCG@20  & 0.0592 & \textbf{0.0615} & 3.89\% & 0.0582 & \textbf{0.0626} & 7.56\% \\
\specialrule{0.5pt}{1pt}{2pt} 
\multirow{6}{*}{Yelp} 
 & HIT@5 & 0.0238 & \textbf{0.0243} & 2.10\% & 0.0236 & \textbf{0.0246} & 4.24\% \\
 & HIT@10 & 0.0404 & \textbf{0.0409} & 1.24\% & 0.0402 & \textbf{0.0421} & 4.73\% \\
 & HIT@20 & 0.0655 & \textbf{0.0676} & 3.21\% & 0.0663 & \textbf{0.0687} & 3.62\% \\
 & NDCG@5 & 0.0150 & \textbf{0.0155} & 3.33\% & 0.0150 & \textbf{0.0155} & 3.33\% \\
 & NDCG@10  & 0.0204 & \textbf{0.0208} & 1.96\% & 0.0202 & \textbf{0.0212} & 4.95\% \\
 & NDCG@20  & 0.0266 & \textbf{0.0275} & 3.38\% & 0.0268 & \textbf{0.0279} & 4.10\% \\
\specialrule{1pt}{1pt}{1pt} 
    \end{tabular}%
    }
    \caption{Performance improvements with our method applied to other frameworks on 4 datasets. The ``+FENRec" notation indicates the integration of our method. ``improv." represents the improvement over the original methods.}
    
    \label{tab:on_others}
\end{table}

Here, we examine the integration of our FENRec model into various contrastive learning frameworks, specifically CL4SRec and DuoRec. As demonstrated in \cref{tab:on_others}, FENRec significantly enhances the performance of contrastive learning frameworks for sequential recommendation (SR) across most evaluated metrics, underscoring its versatility and effectiveness in diverse SR systems. Notably, the integration of FENRec with DuoRec yields more pronounced improvements compared to CL4SRec. This discrepancy can largely be attributed to the nature of the augmentation strategies employed by each framework. CL4SRec relies on methods such as masking and cropping to generate positive pairs. While effective, these strategies can sometimes alter the underlying user intent and diminish semantic similarity, potentially leading to less effective learning of user preferences. Conversely, DuoRec preserves user intent more effectively by aligning sequences that share the same subsequent item as positive pairs. This method enhances the semantic coherence between the pairs, thereby amplifying the benefits when combined with FENRec. This analysis highlights the critical role of augmentation strategies in optimizing the efficacy of contrastive learning frameworks for SR, particularly when augmented with our FENRec model.

\subsection{Ablation Study}
\begin{table}[t]
    \centering
    \footnotesize
    \resizebox{\columnwidth}{!}{%
    % \begin{tabular}{c|l|ccccccc|c|c}
    \begin{tabular}{ll| cccccc}
    \specialrule{1pt}{1pt}{2pt}
    % \multicolumn{2}{c|}{Dataset} & \multicolumn{6}{c|}{\texttt{Electricity}}&\multicolumn{6}{c}{\texttt{Energy}} \\
    Dataset & Method & HIT@5   & HIT@10   & HIT@20   & NDCG@5   & NDCG@10   & NDCG@20    \\
    \specialrule{1pt}{1pt}{2pt}

\multirow{3}{*}{Sports} 
 & FENRec & \textbf{0.0431}  & \textbf{0.0621} & \textbf{0.0890} & \textbf{0.0299} & \textbf{0.0361} & \textbf{0.0429} \\
 & FENRec - S & 0.0403  & 0.0578 & 0.0806 & 0.0282 & 0.0339 & 0.0396 \\
 & FENRec - N & 0.0422  & 0.0616 & 0.0869 & 0.0292 & 0.0354 & 0.0417 \\ 
\specialrule{0.5pt}{1pt}{2pt} 
\multirow{3}{*}{Beauty} 
 & FENRec & \textbf{0.0728}  & \textbf{0.1019} & \textbf{0.1393} & \textbf{0.0514} & \textbf{0.0608} & \textbf{0.0702} \\
 & FENRec - S & 0.0693  & 0.0968 & 0.1317 & 0.0496 & 0.0585 & 0.0673 \\
 & FENRec - N & 0.0707  & 0.0991 & 0.1363 & 0.0501 & 0.0593 & 0.0687 \\
\specialrule{0.5pt}{1pt}{2pt} 
\multirow{3}{*}{Toys} 
 & FENRec & \textbf{0.0818}  & \textbf{0.1109} & \textbf{0.1462} & \textbf{0.0592} & \textbf{0.0686} & \textbf{0.0775} \\
 & FENRec - S & 0.0789  & 0.1055 & 0.1392 & 0.0573 & 0.0659 & 0.0744 \\
 & FENRec - N & 0.0806  & 0.1092 & 0.1437 & 0.0582 & 0.0675 & 0.0762 \\
\specialrule{0.5pt}{1pt}{2pt} 
\multirow{3}{*}{Yelp} 
 & FENRec & \textbf{0.0286}  & \textbf{0.0485} & \textbf{0.0776} & \textbf{0.0182} & \textbf{0.0246} & \textbf{0.0319} \\
 & FENRec - S & 0.0275  & 0.0455 & 0.0740 & 0.0174 & 0.0232 & 0.0304 \\
 & FENRec - N & 0.0268  & 0.0454 & 0.0738 & 0.0170 & 0.0230 & 0.0301 \\
\specialrule{1pt}{1pt}{1pt} 
    \end{tabular}%
    }
    \caption{Ablation study results across all 4 datasets. S denotes Time-Dependent Soft Labeling, and N denotes Enduring Hard Negatives Incorporation. The chart shows performance drops when either S or N is removed.}
    \label{tab:ablation_study}
\end{table}

To evaluate the critical role of each component within the FENRec framework, we executed an ablation study. The results, detailed in \cref{tab:ablation_study}, underscore the individual contributions of each element to the overall effectiveness of the model. Notably, the removal of Time-Dependent Soft Labeling from FENRec resulted in a larger performance decline compared to the exclusion of Enduring Hard Negatives Incorporation. This impact probably stems from the role of time-dependent soft labels in enhancing the model's capacity to handle uncertainty and capture potential user interest. While the integration of enduring hard negatives is pivotal, it primarily assists the model in refining the accuracy of items that users might seem interested in but are not the main targets of prediction. However, this approach carries the inherent risk of incorrectly dismissing genuinely appealing items as irrelevant, leading to potential false negatives. Moreover, soft labels effectively counterbalance this risk by allowing the model to better capture and represent user interest ambiguity, thus optimizing the effectiveness of hard negatives. This synergy underscores the necessity of incorporating both components to achieve optimal performance in FENRec, as shown in our ablation study findings.

\subsection{Analysis}
\subsubsection{Robustness Across Varying Sequence Lengths.}
\begin{figure}[htbp]
    \centering
    \begin{minipage}[b]{0.49\columnwidth}
        \centering
        \includegraphics[width=\linewidth]{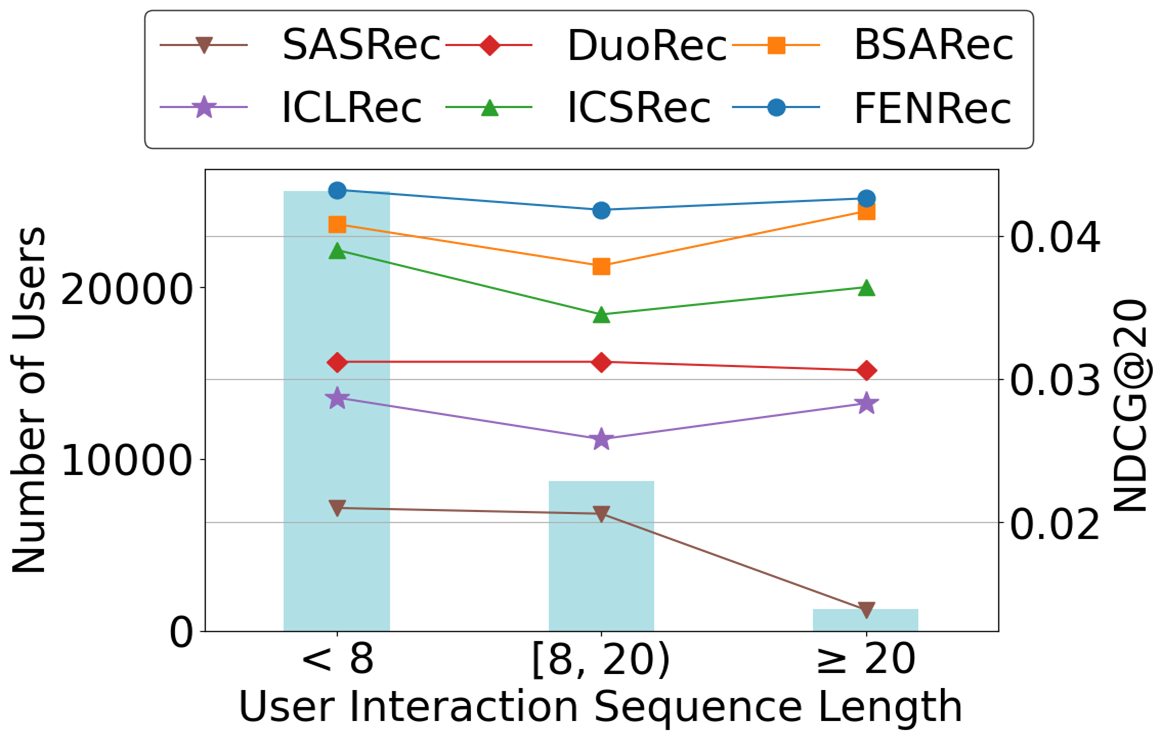}
        \subcaption{Sports}\label{fig:Sports_seqlen}
    \end{minipage}
    \hfill
    \begin{minipage}[b]{0.49\columnwidth}
        \centering
        \includegraphics[width=\linewidth]{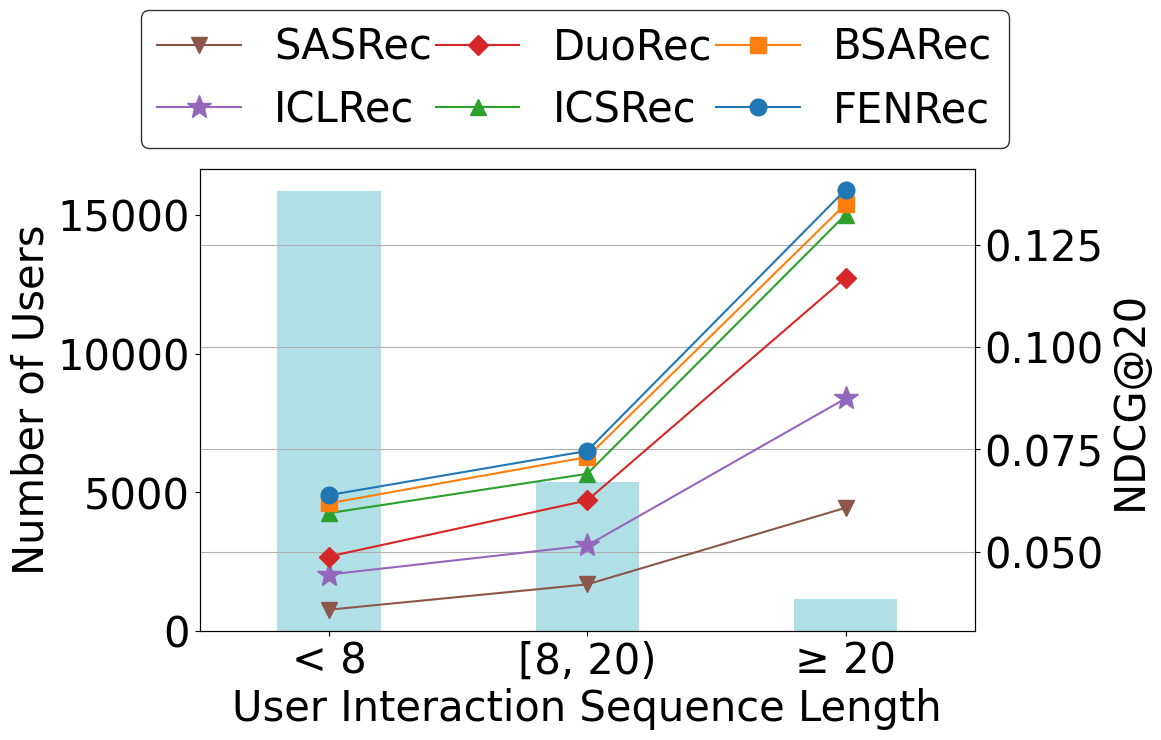}
        \subcaption{Beauty}\label{fig:Beauty_seqlen}
    \end{minipage}
    \vfill
    \begin{minipage}[b]{0.49\columnwidth}
        \centering
        \includegraphics[width=\linewidth]{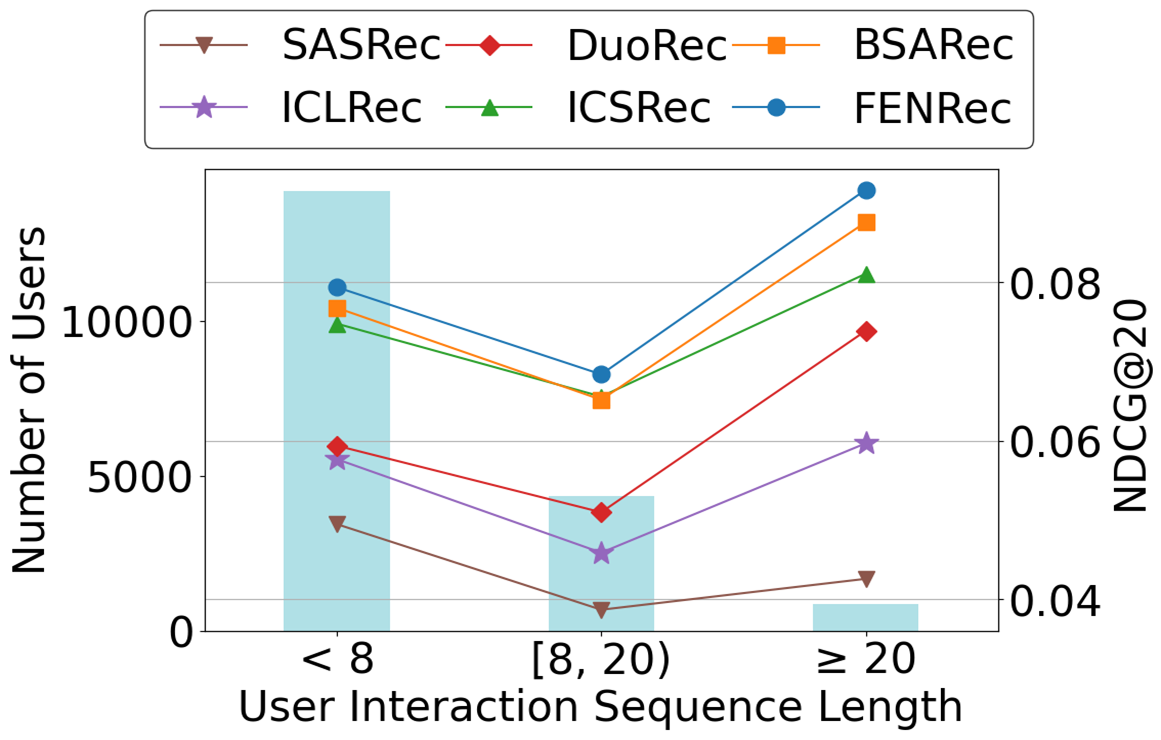} 
        \subcaption{Toys}\label{fig:Toys_seqlen}
    \end{minipage}
    \hfill
    \begin{minipage}[b]{0.49\columnwidth}
        \centering
        \includegraphics[width=\linewidth]{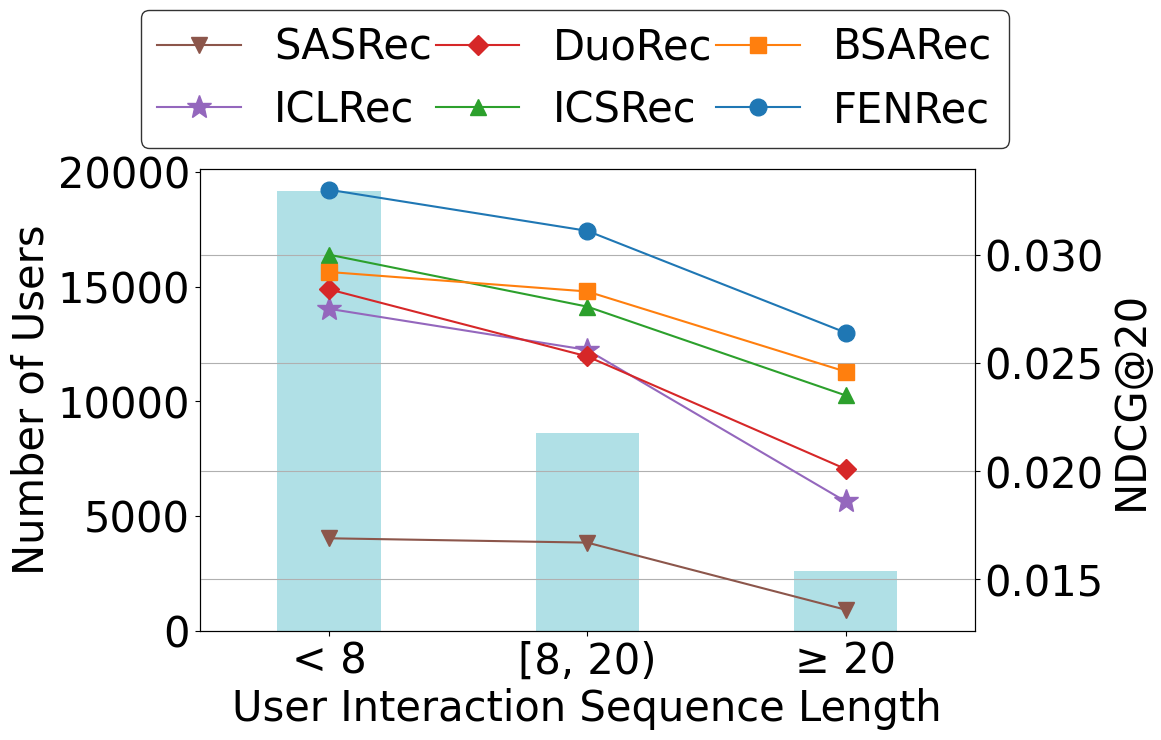} 
        \subcaption{Yelp}\label{fig:Yelp_seqlen}
    \end{minipage}
    \caption[Comparison of Model Performance Across Different User Interaction Sequence Lengths.]{Comparison of model performance across different user interaction sequence lengths. The bar chart represents the number of users, and the line chart indicates NDCG@20.}\label{fig:seqlen_exp}
\end{figure}

We tested the robustness of FENRec by evaluating its performance across user groups with different interaction sequence lengths (i.e., fewer than 8, 8 to 20, and more than 20 interactions). As shown in \cref{fig:seqlen_exp}, contrastive learning methods enhance the performance of the backbone (i.e., SASRec) across all sequence lengths, highlighting their effectiveness. On Amazon datasets (i.e., Sports, Beauty, and Toys), contrastive learning methods show greater performance improvements for users with longer interaction sequences, while on the Yelp dataset, shorter sequences benefit more. This may be due to users with longer sequences having a lower average number of interactions per item within their sequences on the Yelp dataset. For the statistics, readers can refer to Appendix.
Overall, our method consistently outperforms baseline methods across all sequence lengths, including scenarios with cold start issues (i.e., Beauty), demonstrating the robustness of FENRec.

\subsubsection{Model Discriminative Capability Analysis.}

\begin{figure}[htbp]
    \centering
    \begin{minipage}[b]{0.48\columnwidth}
        \centering
        \includegraphics[width=\linewidth]{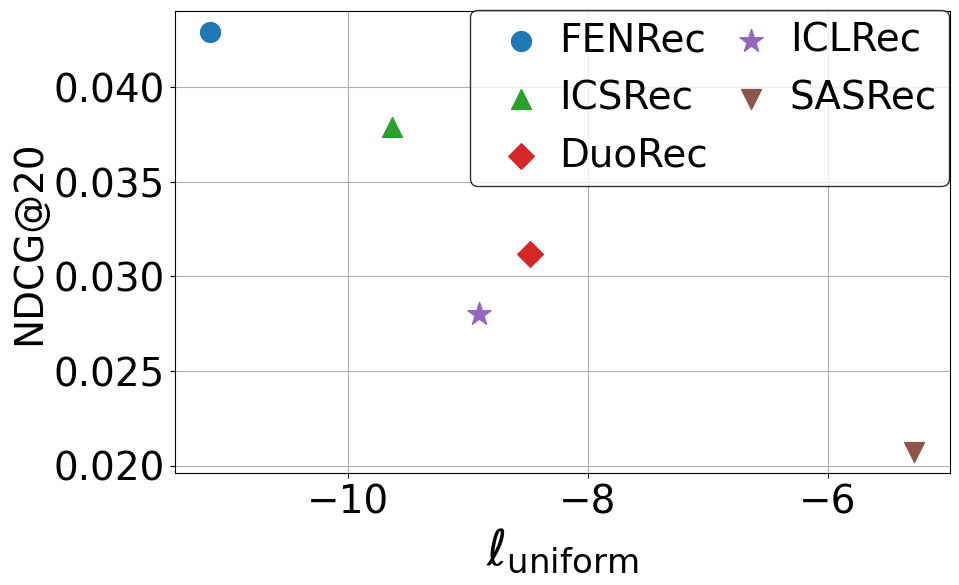}
        \subcaption{Sports}\label{fig:Sports_uni}
    \end{minipage}
    \hfill
    \begin{minipage}[b]{0.48\columnwidth}
        \centering
        \includegraphics[width=\linewidth]{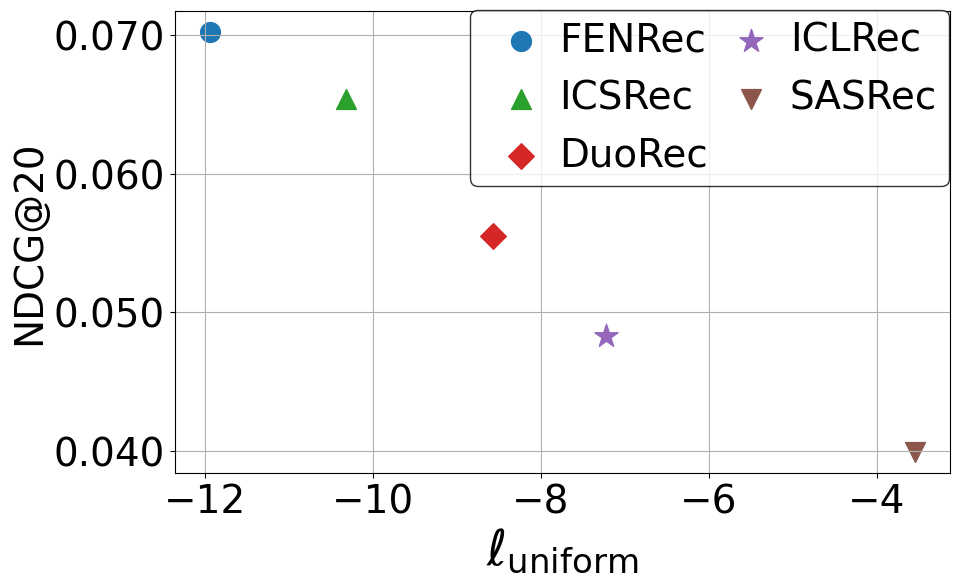}
        \subcaption{Beauty}\label{fig:Beauty_uni}
    \end{minipage}
    \vfill
    \begin{minipage}[b]{0.48\columnwidth}
        \centering
        \includegraphics[width=\linewidth]{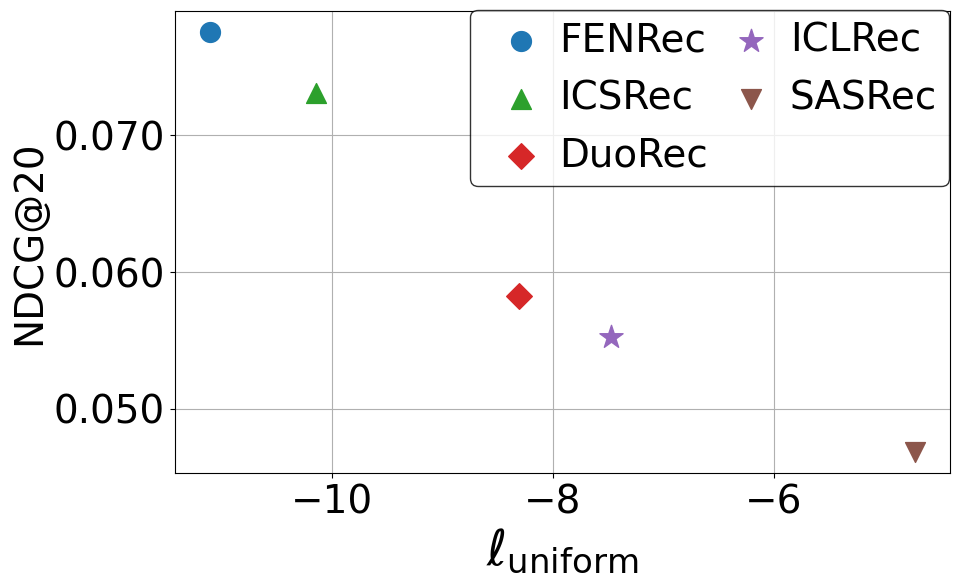} 
        \subcaption{Toys}\label{fig:Toys_uni}
    \end{minipage}
    \hfill
    \begin{minipage}[b]{0.48\columnwidth}
        \centering
        \includegraphics[width=\linewidth]{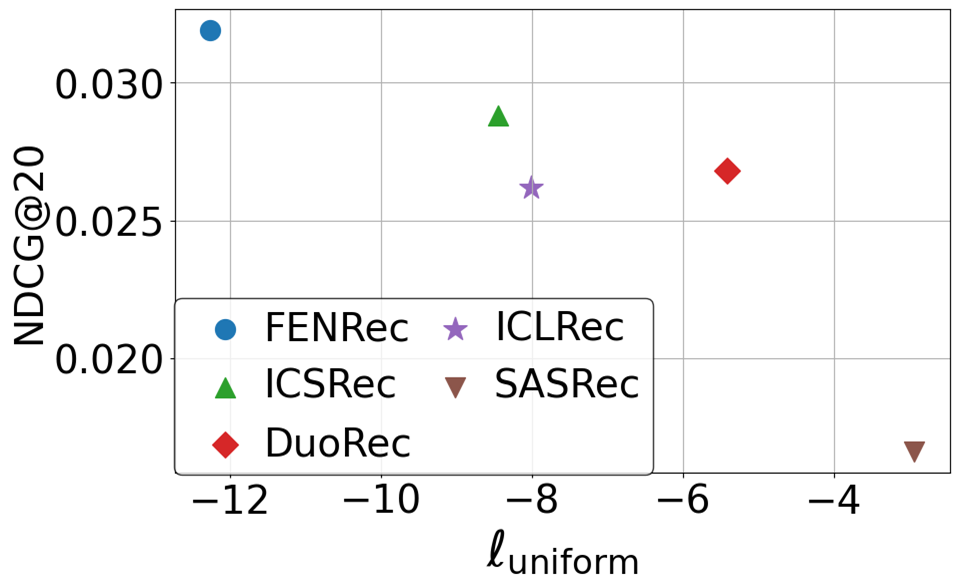} 
        \subcaption{Yelp}\label{fig:Yelp_uni}
    \end{minipage}
    \caption{Comparsion of $\ell_{\text{uniform}}$ and performance.}\label{fig:uniform}
\end{figure}
In this section, we evaluate the model's discriminative capability using item uniformity, which indicates the model's ability to differentiate between semantically similar items. Lower uniformity suggests higher semantic similarities, indicating that the model may struggle to distinguish between items, while higher uniformity indicates a more even distribution of item embeddings, enabling effective differentiation. 
We compare item uniformity across various contrastive learning methods and their backbone, SASRec. The uniformity loss $\ell_{\text{uniform}}$  measures data uniformity. It is calculated using the definition in DuoRec~\cite{qiu2022duorec}. A smaller $\ell_{\text{uniform}}$ represents a more uniform data distribution. As shown in \cref{fig:uniform}, there is a negative correlation between $\ell_{\text{uniform}}$ and NDCG@20 since lower uniformity indicates that the model has not learned sufficiently discriminative features, leading to lower performance. Results demonstrate that FENRec helps increase item uniformity and performance, indicating its effectiveness in helping the model learn discriminative features. This improvement is likely because both components aid in learning finer user preferences. Hard negatives help the model differentiate similar preferences, while soft labels provide finer labels, enhancing the discriminative power of the learned representations and thereby improving the model's comprehension of item semantics.

\subsection{Hyperparameter Tuning}
\begin{figure}[htbp]
    \centering
    \begin{minipage}[b]{0.45\columnwidth}
        \centering
        \includegraphics[width=\linewidth]{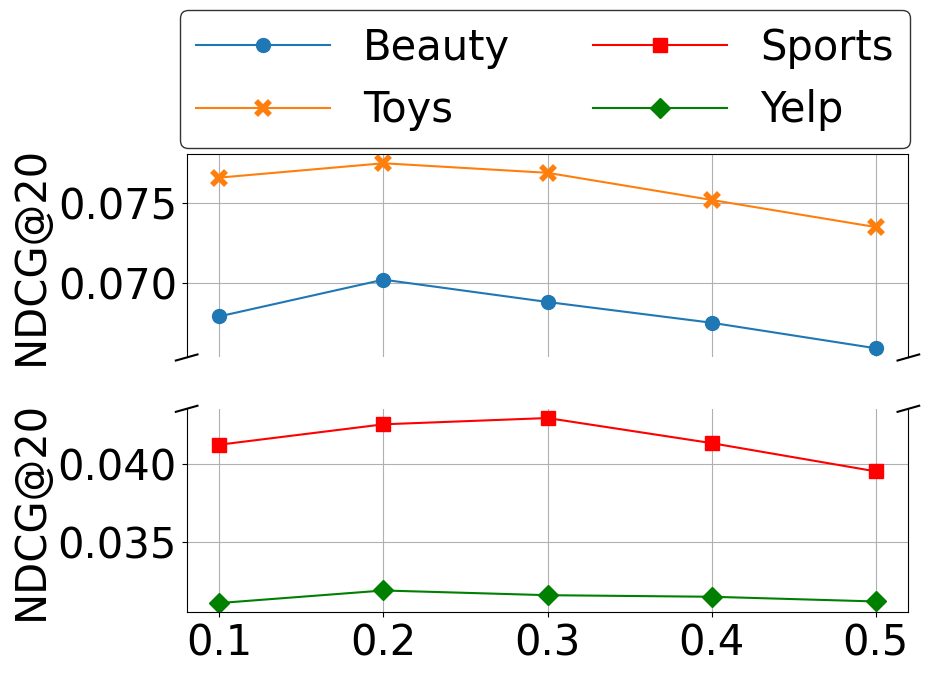}
        \subcaption{$\gamma$}\label{fig:gamma_NDCG}
    \end{minipage}
    \hfill
    \begin{minipage}[b]{0.45\columnwidth}
        \centering
        \includegraphics[width=\linewidth]{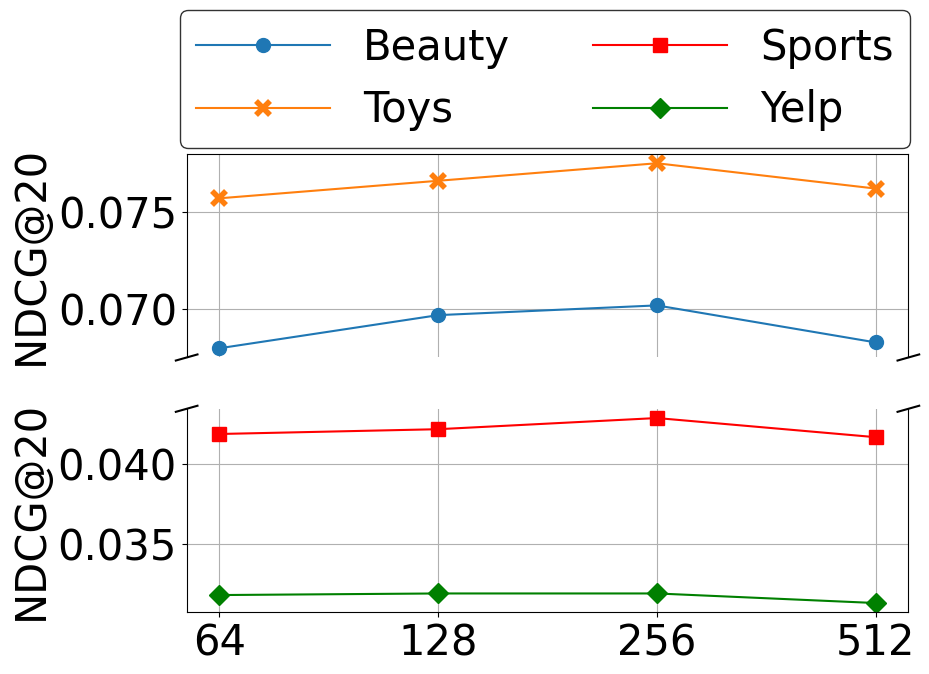}
        \subcaption{Batch size}\label{fig:batch_size_NDCG}
    \end{minipage}
    \caption{Performance of FENRec w.r.t. different hyperparameters.}
    \label{fig:tuning}
\end{figure}

We tune each hyperparameter individually while keeping others optimal to observe performance changes, as shown in \cref{fig:tuning}. First, we adjust $\gamma$. A higher $\gamma$ creates a smoother distribution of soft labels, while a lower $\gamma$ makes it more concentrated.
FENRec shows the best performance when $\gamma$ is 0.2 or 0.3, indicating that a low $\gamma$ can make the model overly confident and a high $\gamma$ can make it unsure about the next item. Next, we assess the impact of batch size, finding FENRec performs best at 256. A smaller batch size lacks sufficient negative samples, while a larger one increases the risk of false negatives, degrading performance. This risk is amplified as FENRec generates hard negatives by mixing anchor and in-batch negatives, causing false negatives to be pushed further away from positives. Due to the space constraint, results of other hyperparameters are in Appendix.

\section{Conclusion}
In this paper, we tackled the issue of data sparsity in sequential recommendation systems by introducing FENRec, an innovative approach that integrates Time-Dependent Soft Labeling and Enduring Hard Negatives Incorporation within contrastive learning frameworks. By utilizing future interactions in Time-Dependent Soft Labeling, our method effectively captures finer user preferences. Additionally, Enduring Hard Negatives Incorporation ensures the model learns from more challenging samples, enhancing its ability to learn discriminative features. Extensive experiments on four benchmark datasets highlight FENRec's effectiveness in enhancing performance and improving the differentiation of user preferences. For future work, we would like to incorporate auxiliary time information to improve the performance.

\section*{Acknowledgments}
  This work is partially supported by the National Science and Technology Council, Taiwan under Grants NSTC-112-2221-E-A49-059-MY3 and NSTC-112-2221-E-A49-094-MY3.

\bibliography{aaai25}

\clearpage
\begin{appendices}

\renewcommand{\thefigure}{\Alph{figure}}
\renewcommand{\thetable}{\Alph{table}}

\appendix

\setcounter{equation}{16}

\section{Method Details}
\subsection{Enduring Hard Negatives Incorporation}
The following demonstrates that the enduring hard negatives $\mathbf{h}^{-}_{i,j}$ we proposed satisfy $\mathbf{h}^{-}_{i,j} \cdot \mathbf{h}_{i} \geq \mathbf{h}_{j} \cdot \mathbf{h}_{i}$, where $\mathbf{h}_i$ is the anchor sample and $\mathbf{h}_j$  represents the original negative samples.

\newtheorem{lemma}{Lemma}

\begin{lemma}
\label{lem:vector-angle}
Let $ \mathbf{x} $ and $ \mathbf{y} $ be non-zero vectors (i.e., $ \|\mathbf{x}\|_2 \neq 0 $, $ \|\mathbf{y}\|_2 \neq 0 $) and $ \mathbf{x} \neq -\mathbf{y} $. Define $ \mathbf{z} = \mathbf{x} + \mathbf{y} $. Then the cosine of the angle between $ \mathbf{x} $ and $ \mathbf{z} $ is greater than or equal to the cosine of the angle between $ \mathbf{x} $ and $ \mathbf{y} $. Formally, we have:

\[
\cos(\theta_{\mathbf{x}\mathbf{z}}) \geq \cos(\theta_{\mathbf{x}\mathbf{y}}).
\]

\end{lemma}

\begin{proof}
To establish the result, we first decompose $\mathbf{y}$ into two components: the projection of $\mathbf{y}$ onto $\mathbf{x}$, denoted as $\mathbf{y}_x$, which is in the same direction as $\mathbf{x}$, and the component of $\mathbf{y}$ that is orthogonal to $\mathbf{x}$, denoted as $\mathbf{y}_{x^\perp}$. Thus, we can express $\mathbf{y}$ as:
\begin{align}
\mathbf{y} = \mathbf{y}_x + \mathbf{y}_{x^\perp}.
\end{align}
Consequently, the resultant vector $\mathbf{z}$ can be written as:
\begin{align}
\mathbf{z} = \mathbf{x} + \mathbf{y} = \mathbf{x} + \mathbf{y}_x + \mathbf{y}_{x^\perp}.
\end{align}
We now consider two cases: 1) $\mathbf{x}$ and $\mathbf{y}$ are collinear and 2) $\mathbf{x}$ and $\mathbf{y}$ are not collinear
\begin{enumerate}[label=\arabic*)]
    \item When $\mathbf{x}$ and $\mathbf{y}$ are collinear, the projection of $\mathbf{y}$ onto $\mathbf{x}$, denoted as $\mathbf{y}_x$, equals $\mathbf{y}$, and the orthogonal component $\mathbf{y}_{x^\perp}$ is zero. This results in either of the following scenarios\footnote{We constrain the angles $\theta_{\mathbf{x}\mathbf{z}}$ and $\theta_{\mathbf{x}\mathbf{y}}$ to the range $ -\pi \leq \theta < \pi $.}: a) $\theta_{\mathbf{x}\mathbf{z}} = \theta_{\mathbf{x}\mathbf{y}} = 0$ or b) $\theta_{\mathbf{x}\mathbf{z}} = \theta_{\mathbf{x}\mathbf{y}} + \pi = 0$.
    \begin{enumerate}[label=\alph*), leftmargin=2em]
        \item Since $\theta_{\mathbf{x}\mathbf{z}} = \theta_{\mathbf{x}\mathbf{y}}$, we have $\cos(\theta_{\mathbf{x}\mathbf{z}}) = \cos(\theta_{\mathbf{x}\mathbf{y}})$. 
        \item $\theta_{\mathbf{x}\mathbf{z}} = 0$ and $\theta_{\mathbf{x}\mathbf{y}} = - \pi$, which implies $\cos(\theta_{\mathbf{x}\mathbf{z}}) > \cos(\theta_{\mathbf{x}\mathbf{y}})$. 
    \end{enumerate}
    Therefore, we conclude that $\cos(\theta_{\mathbf{x}\mathbf{z}}) \geq \cos(\theta_{\mathbf{x}\mathbf{y}})$.
    \item When $\mathbf{x}$ and $\mathbf{y}$ are not collinear, the vector $\mathbf{y}$ has a non-zero orthogonal component $\mathbf{y}_{x^\perp}$. As a result, $\mathbf{z}$, which includes both $\mathbf{x}$ and the aligned component $\mathbf{y}_x$ (i.e., the projection of $\mathbf{y}$ onto $\mathbf{x}$), is closer in direction to $\mathbf{x}$ than $\mathbf{y}$, making the angle $\theta_{\mathbf{x}\mathbf{z}}$ closer to 0 compared to $\theta_{\mathbf{x}\mathbf{y}}$. This implies $\cos(\theta_{\mathbf{x}\mathbf{z}}) > \cos(\theta_{\mathbf{x}\mathbf{y}})$.
\end{enumerate}
Considering both cases, we can conclude that:
\begin{align}
\cos(\theta_{\mathbf{x}\mathbf{z}}) \geq \cos(\theta_{\mathbf{x}\mathbf{y}}).
\end{align}
\end{proof}

\begin{lemma}
Given non-zero vectors \( \mathbf{h}_i \) and \( \mathbf{h}_j \), and let \( 0 < \lambda < 1 \), define
\[
    \mathbf{\tilde{h}}_{i} = \frac{\mathbf{h}_{i}}{\|\mathbf{h}_{i}\|_2},
\]
\[
    \mathbf{\tilde{h}}_{j} = \frac{\mathbf{h}_{j}}{\|\mathbf{h}_{j}\|_2},
\]
\[
    \mathbf{h}^{-}_{i,j} = \frac{\lambda \mathbf{\tilde{h}}_{i} + (1 - \lambda) \mathbf{\tilde{h}}_{j}}{\|\lambda \mathbf{\tilde{h}}_{i} + (1 - \lambda) \mathbf{\tilde{h}}_{j}\|_2} \cdot \|\mathbf{h}_{j}\|_2,
\]
and 
\[
    \mathbf{{h}}_{i} \neq \beta\mathbf{{h}}_{j}, \text{where } \beta \in \mathbb{R},
\]
we have:
\[
    \mathbf{h}^{-}_{i,j} \cdot \mathbf{h}_{i} > \mathbf{h}_{j} \cdot \mathbf{h}_{i}.
\]
\end{lemma}

\begin{proof}
To show that $\mathbf{h}^{-}_{i,j} \cdot \mathbf{h}_{i} > \mathbf{h}_{j} \cdot \mathbf{h}_{i}$, we begin by expressing the dot products in terms of the norms and cosines of the angles:
\begin{align}
\mathbf{h}^{-}_{i,j} \cdot \mathbf{h}_{i} &= \|\mathbf{h}^{-}_{i,j}\| \|\mathbf{h}_{i}\| \cos(\theta_{\mathbf{h}^{-}_{i,j}, \mathbf{h}_i}),
\end{align}
and
\begin{align}
\mathbf{h}_{j} \cdot \mathbf{h}_{i} &= \|\mathbf{h}_{j}\| \|\mathbf{h}_{i}\| \cos(\theta_{\mathbf{h}_j, \mathbf{h}_i}).
\end{align}
The inequality $\mathbf{h}^{-}_{i,j} \cdot \mathbf{h}_{i} > \mathbf{h}_{j} \cdot \mathbf{h}_{i}$ can be rewritten as:
\begin{align}
\|\mathbf{h}^{-}_{i,j}\| \|\mathbf{h}_{i}\| \cos(\theta_{\mathbf{h}^{-}_{i,j}, \mathbf{h}_i}) > \|\mathbf{h}_{j}\| \|\mathbf{h}_{i}\| \cos(\theta_{\mathbf{h}_j, \mathbf{h}_i}).
\end{align}
Since \( \|\mathbf{h}_{i}\|, \|\mathbf{h}^{-}_{i,j}\| \), and $\|\mathbf{h}_j\|$ are positive and \( \|\mathbf{h}^{-}_{i,j}\| = \|\mathbf{h}_j\| \), the inequality simplifies to:
\begin{align}
\cos(\theta_{\mathbf{h}^{-}_{i,j}, \mathbf{h}_i}) > \cos(\theta_{\mathbf{h}_j, \mathbf{h}_i}),
\end{align}
which directly compares the cosines of the angles between the vectors.
Next, considering that $0 < \lambda < 1$, we observe that:
\begin{align}
\cos(\theta_{\mathbf{h}^{-}_{i,j}, \mathbf{h}_i}) > \cos(\theta_{\mathbf{h}_j, \mathbf{h}_i})
\label{eq:cos1}
\end{align}
holds if and only if:
\begin{align}
\cos(\theta_{\mathbf{\hat{h}}^{-}_{i,j}, \lambda\mathbf{\tilde{h}}_i}) > \cos(\theta_{(1-\lambda)\mathbf{\tilde{h}}_j, \lambda\mathbf{\tilde{h}}_i}),
\label{eq:cos2}
\end{align}
where:
\begin{align}
\mathbf{\hat{h}}^{-}_{i,j} &= \lambda \mathbf{\tilde{h}}_i + (1 - \lambda) \mathbf{\tilde{h}}_j.
\end{align}
The equivalence between Equation~\ref{eq:cos1} and Equation~\ref{eq:cos2} holds because the cosine of the angle between two vectors depends solely on the direction of the vectors and is independent of their norms, provided the vectors are non-zero. In other words, normalizing the vectors or scaling them by a positive scalar does not change the cosine of the angle between them.

By Lemma~\ref{lem:vector-angle}, since $\mathbf{\hat{h}}^{-}_{i,j} = \lambda \mathbf{\tilde{h}}_i + (1 - \lambda) \mathbf{\tilde{h}}_j$ is the linear combination of $\lambda \mathbf{\tilde{h}}_i$ and $(1 - \lambda) \mathbf{\tilde{h}}_j$, and given that $\mathbf{h}_i \neq \beta \mathbf{h}_j$ for $\beta \in \mathbb{R}$, which implies $\mathbf{\tilde{h}}_i \neq -\mathbf{\tilde{h}}_j$, we have:
\begin{align}
\cos(\theta_{\mathbf{\hat{h}}^{-}_{i,j}, \lambda\mathbf{\tilde{h}}_i}) \geq \cos(\theta_{(1-\lambda)\mathbf{\tilde{h}}_j, \lambda\mathbf{\tilde{h}}_i}),
\end{align}
with equality only if $\mathbf{\tilde{h}}_i$ and $\mathbf{\tilde{h}}_j$ are collinear, which contradicts the assumption that \(\mathbf{h}_i \neq \beta \mathbf{h}_j\) for $\beta \in \mathbb{R}$. Thus, the inequality is strict:
\begin{align}
\cos(\theta_{\mathbf{\hat{h}}^{-}_{i,j}, \lambda\mathbf{\tilde{h}}_i}) > \cos(\theta_{(1-\lambda)\mathbf{\tilde{h}}_j, \lambda\mathbf{\tilde{h}}_i}).
\end{align}
Therefore:
\begin{align}
\mathbf{h}^{-}_{i,j} \cdot \mathbf{h}_{i} &> \mathbf{h}_{j} \cdot \mathbf{h}_{i}.
\end{align}
\end{proof}

Finally, the equality $\mathbf{h}^{-}_{i,j} \cdot \mathbf{h}_{i} = \mathbf{h}_{j} \cdot \mathbf{h}_{i}$ holds only if $\mathbf{h}_i$ and $\mathbf{h}_j$ are collinear. Since $\mathbf{h}_i$ and $\mathbf{h}_j$ are 64-dimensional vectors in the context of model training, the probability of the equality $\mathbf{h}^{-}_{i,j} \cdot \mathbf{h}_{i} = \mathbf{h}_{j} \cdot \mathbf{h}_{i}$ holding is exceedingly small. In high-dimensional spaces, the chance of two randomly chosen vectors being exactly collinear is almost zero. Therefore, in practice, $\mathbf{h}^{-}_{i,j} \cdot \mathbf{h}_{i}$ is greater than $\mathbf{h}_{j} \cdot \mathbf{h}_{i}$. 

We demonstrated that the enduring hard negatives $\mathbf{h}^{-}_{i,j}$ we proposed satisfy $\mathbf{h}^{-}_{i,j} \cdot \mathbf{h}_{i} \geq \mathbf{h}_{j} \cdot \mathbf{h}_{i}$. Morover, in most cases, $\mathbf{h}^{-}_{i,j} \cdot \mathbf{h}_{i}$ is likely to be greater than $\mathbf{h}_{j} \cdot \mathbf{h}_{i}$ with high probability.

\section{Experiment Details}
\subsection{Experimental Setup}
\subsubsection{Datasets.}
We conduct our experiments on the Amazon\footnote{\url{http://jmcauley.ucsd.edu/data/amazon/}} Sports, Amazon Beauty, Amazon Toys, and Yelp\footnote{\url{https://www.yelp.com/dataset}} datasets.

\begin{table}[]
    \centering
    \begin{tabular}{|c|c|c|c|}

    \hline
    \multirow{2}{*}{Dataset} & \multicolumn{3}{c|}{\textbf{User Interaction Sequence Length}} \\ \cline{2-4} 
    & \textbf{$<$ 8} & \textbf{[8, 20)} & \textbf{$\geq$ 20} \\ \hline
    \textbf{Sports} & 45.8798 & 48.4119 & 48.1976 \\
    \hline
    \textbf{Beauty} & 37.8374 & 42.1797 & 49.4946 \\
    \hline
    \textbf{Toys} & 25.6032 & 25.7388 & 27.3176 \\
    \hline
    \textbf{Yelp} & 34.4817 & 32.8905 & 30.1013 \\
    \hline
    \end{tabular}
    \caption{Average number of interactions per item within user interaction sequences of different lengths across each dataset.}
    \label{tab:interactions_per_item}
\end{table}

The statistics of the average number of interactions per item within user interaction sequences of different lengths across each dataset, which is mentioned in the main text, is shown in Table~\ref{tab:interactions_per_item}. The interactions per item in a user interaction sequence is calculated by the following:
\begin{align}
    \text{Interactions per Item within }\mathcal{S}^{u} = \frac{I_{\text{total}}^u}{|\mathcal{S}^{u}|},
\end{align}
where
$I_{\text{total}}^u$ denotes the \textbf{Total Number of Interactions} in the training data across all items in the user interaction sequence for a specific user $u$.
$|\mathcal{S}^{u}|$ represents the \textbf{Length of the User Interaction Sequence} for a specific user $u$.
\subsubsection{Baselines.} We compare our method, FENRec, with state-of-the-art sequential recommendation (SR) approaches, broadly divided into 3 categories:
\begin{itemize}
    \item \textbf{General sequential models}: 
    \textbf{Caser}~\cite{tang2018caser} utilize convolutional neural networks (CNNs) in SR. \textbf{GRURec} employs Recurrent Neural Networks (RNNs) %as the encoder 
    in SR. \textbf{SASRec}~\cite{kang2018sasrec} leverages the transformer-based model for SR.
    %to capture long-range dependencies within user behavior sequences. 
    \textbf{LRURec}~\cite{yue2024lrurec} first introduces Linear Recurrent Units in SR to improve the efficiency of SR models. %and shows its great efficacy.
    \textbf{BSARec}~\cite{shin2024bsarec} introduces an attentive inductive bias using the Fourier transform to address the over-smoothing problem in the self-attention of the SR model.
    
    \item \textbf{Sequential models with self-supervised learning}: 
    \textbf{BERT4Rec}~\cite{sun2019bert4rec} uses BERT and the masked item prediction task for SR. \textbf{MAERec}~\cite{ye2023graph} introduces a Graph Masked Autoencoder to distill global item transition information for self-supervised augmentation adaptively. 
    Some works introduce contrastive learning.
    %, which enhances the effectiveness of learned representations by utilizing information from augmented views.
    \textbf{CL4SRec}~\cite{xie2022cl4srec} uses data augmentations such as masking and cropping to generate augmented views for contrastive learning in SR. \textbf{CoSeRec}~\cite{liu2021coserec} further introduces two informative augmentation operators for contrastive learning. 
    \textbf{CBiT}~\cite{du2022CBiT} uses contrastive learning with a BERT architecture, employing cloze task and dropout masks to generate positive samples.
    \textbf{DuoRec}~\cite{qiu2022duorec} utilize model-level augmentation and regard user sequence with the same next item as augmented positive view. \textbf{ICLRec}~\cite{chen2022iclrec} employs intent contrastive learning and models latent user intents through clustering. \textbf{ICSRec}~\cite{qin2024icsrec} introduces intent contrastive learning within subsequences, utilizing coarse-grain and fine-grain intent contrastive learning methods. Note that CL4SRec, CoSeRec, ICLRec, DuoRec, and ICSRec all share the same backbone, SASRec.

    \item \textbf{Sequential models with label smoothness}: 
    \textbf{MVS}~\cite{zhou2023MVS} enhances SR models by introducing smoothness into data representation and model learning, using complementary models to enrich context and label representations. 
    %Similar to the contrastive learning methods, 
    Here, we use MVS with SASRec as the backbone model for the experiments.
\end{itemize}

\subsubsection{Implementation Details.}
The implementation details of our method, FENRec, are in the main text. For the baselines, we used public implementations for Caser\footnote{\url{https://github.com/graytowne/caser_pytorch}}, SASRec\footnote{\url{https://github.com/pmixer/SASRec.pytorch}}, and GRU4Rec\footnote{\url{https://github.com/yehjin-shin/BSARec}}.
For BERT4Rec, CL4SRec, MAERec, and DuoRec, we utilize the implementation provided by SSLRec\footnote{\url{https://github.com/HKUDS/SSLRec}}~\cite{ren2024sslrec}. 
For MVS, ICLRec, ICSRec, LRURec, BSARec, and CoSeRec, we employ the code provided by the author. We configure
the embedding dimension to 64. The maximum user
sequence length is set to 50 across all methods. Shorter sequences are padded and longer ones are truncated. We adjust the batch size to 256, although for the MVS system, we use a reduced batch size of 64 for the Sports and Yelp datasets during training due to memory constraints. All other hyper-parameters for each baseline are set following the suggestions in the original papers, and we report each baseline’s performance under its optimal settings. We conduct the experiments using an NVIDIA RTX 3090 GPU, and the code implementations are in PyTorch.

\begin{table*}[t]
    \centering
    \footnotesize
    \resizebox{\textwidth}{!}{%
    % \begin{tabular}{c|l|ccccccc|c|c}
    \begin{tabular}{ll| cccccccc}
    \specialrule{1pt}{1pt}{2pt}
    % \multicolumn{2}{c|}{Dataset} & \multicolumn{6}{c|}{\texttt{Electricity}}&\multicolumn{6}{c}{\texttt{Energy}} \\
    Dataset & Metric & GRU4Rec & Caser  & LRURec & SASRec  & BSARec & BERT4Rec & MAERec & CBiT    \\
    \specialrule{1pt}{1pt}{2pt}

\multirow{6}{*}{Sports} 
 & HIT@5 & 0.0116 ± 0.0012 & 0.0123 ± 0.0007 & 0.0389 ± 0.0014 & 0.0189 ± 0.0007 & \underline{0.0400 ± 0.0005} & 0.0264 ± 0.0007 & 0.0285 ± 0.0004 & 0.0235 ± 0.0013 \\
& HIT@10 & 0.0197 ± 0.0013 & 0.0210 ± 0.0006 & 0.0551 ± 0.0014 & 0.0307 ± 0.0014 & \underline{0.0583 ± 0.0005} & 0.0408 ± 0.0006 & 0.0435 ± 0.0003 & 0.0365 ± 0.0007 \\
& HIT@20 & 0.0320 ± 0.0017 & 0.0336 ± 0.0014 & 0.0771 ± 0.0006 & 0.0491 ± 0.0022 & \underline{0.0830 ± 0.0005} & 0.0622 ± 0.0008 & 0.0645 ± 0.0008 & 0.0528 ± 0.0010 \\
& NDCG@5 & 0.0074 ± 0.0009 & 0.0078 ± 0.0005 & 0.0276 ± 0.0011 & 0.0122 ± 0.0004 & \underline{0.0280 ± 0.0005} & 0.0175 ± 0.0005 & 0.0191 ± 0.0003 & 0.0157 ± 0.0005 \\
& NDCG@10 & 0.0100 ± 0.0010 & 0.0105 ± 0.0003 & 0.0329 ± 0.0011 & 0.0161 ± 0.0007 & \underline{0.0339 ± 0.0004} & 0.0221 ± 0.0005 & 0.0239 ± 0.0003 & 0.0198 ± 0.0003 \\
& NDCG@20 & 0.0131 ± 0.0011 & 0.0137 ± 0.0002 & 0.0384 ± 0.0008 & 0.0207 ± 0.0008 & \underline{0.0401 ± 0.0005} & 0.0275 ± 0.0004 & 0.0292 ± 0.0002 & 0.0239 ± 0.0003 \\

\specialrule{0.5pt}{1pt}{2pt} 

\multirow{6}{*}{Beauty} 
& HIT@5 & 0.0188 ± 0.0026 & 0.0234 ± 0.0004 & 0.0671 ± 0.0014 & 0.0359 ± 0.0007 & \underline{0.0707 ± 0.0002} & 0.0489 ± 0.0018 & 0.0557 ± 0.0013 & 0.0612 ± 0.0015 \\
& HIT@10 & 0.0315 ± 0.0034 & 0.0386 ± 0.0002 & 0.0928 ± 0.0014 & 0.0580 ± 0.0012 & \underline{0.0978 ± 0.0009} & 0.0735 ± 0.0027 & 0.0789 ± 0.0025 & 0.0871 ± 0.0029 \\
& HIT@20 & 0.0516 ± 0.0055 & 0.0585 ± 0.0009 & 0.1257 ± 0.0015 & 0.0905 ± 0.0030 & \underline{0.1345 ± 0.0020} & 0.1065 ± 0.0029 & 0.1094 ± 0.0019 & 0.1202 ± 0.0028 \\
& NDCG@5 & 0.0114 ± 0.0016 & 0.0148 ± 0.0002 & 0.0481 ± 0.0010 & 0.0233 ± 0.0007 & \underline{0.0503 ± 0.0001} & 0.0330 ± 0.0017 & 0.0397 ± 0.0009 & 0.0435 ± 0.0013 \\
& NDCG@10 & 0.0154 ± 0.0019 & 0.0197 ± 0.0002 & 0.0564 ± 0.0010 & 0.0304 ± 0.0005 & \underline{0.0590 ± 0.0003} & 0.0409 ± 0.0015 & 0.0472 ± 0.0011 & 0.0518 ± 0.0017 \\
& NDCG@20 & 0.0205 ± 0.0024 & 0.0248 ± 0.0001 & 0.0647 ± 0.0010 & 0.0385 ± 0.0009 & \underline{0.0682 ± 0.0006} & 0.0492 ± 0.0015 & 0.0548 ± 0.0012 & 0.0602 ± 0.0017 \\
\specialrule{0.5pt}{1pt}{2pt} 

\multirow{6}{*}{Toys} 
 & HIT@5 & 0.0164 ± 0.0017 & 0.0180 ± 0.0004 & 0.0707 ± 0.0014 & 0.0481 ± 0.0019 & \underline{0.0792 ± 0.0019} & 0.0476 ± 0.0012 & 0.0589 ± 0.0009 & 0.0632 ± 0.0006 \\
& HIT@10 & 0.0277 ± 0.0027 & 0.0277 ± 0.0006 & 0.0941 ± 0.0005 & 0.0699 ± 0.0015 & \underline{0.1066 ± 0.0015} & 0.0690 ± 0.0025 & 0.0823 ± 0.0011 & 0.0865 ± 0.0006 \\
& HIT@20 & 0.0461 ± 0.0037 & 0.0421 ± 0.0009 & 0.1228 ± 0.0001 & 0.0982 ± 0.0011 & \underline{0.1405 ± 0.0018} & 0.0974 ± 0.0024 & 0.1108 ± 0.0007 & 0.1166 ± 0.0006 \\
& NDCG@5 & 0.0104 ± 0.0011 & 0.0117 ± 0.0006 & 0.0523 ± 0.0008 & 0.0326 ± 0.0015 & \underline{0.0574 ± 0.0013} & 0.0332 ± 0.0013 & 0.0424 ± 0.0006 & 0.0453 ± 0.0005 \\
& NDCG@10 & 0.0140 ± 0.0014 & 0.0149 ± 0.0004 & 0.0598 ± 0.0004 & 0.0396 ± 0.0012 & \underline{0.0662 ± 0.0013} & 0.0401 ± 0.0017 & 0.0499 ± 0.0007 & 0.0529 ± 0.0005 \\
& NDCG@20 & 0.0187 ± 0.0017 & 0.0185 ± 0.0005 & 0.0671 ± 0.0004 & 0.0468 ± 0.0010 & \underline{0.0747 ± 0.0013} & 0.0472 ± 0.0017 & 0.0570 ± 0.0006 & 0.0605 ± 0.0005 \\

\specialrule{0.5pt}{1pt}{2pt} 

\multirow{6}{*}{Yelp} 
& HIT@5 & 0.0129 ± 0.0008 & 0.0137 ± 0.0001 & 0.0240 ± 0.0003 & 0.0147 ± 0.0006 & 0.0252 ± 0.0008 & 0.0215 ± 0.0004 & 0.0255 ± 0.0002 & 0.0164 ± 0.0006 \\
& HIT@10 & 0.0227 ± 0.0017 & 0.0246 ± 0.0003 & 0.0410 ± 0.0005 & 0.0254 ± 0.0009 & \underline{0.0432 ± 0.0017} & 0.0361 ± 0.0009 & 0.0423 ± 0.0003 & 0.0281 ± 0.0015 \\
& HIT@20 & 0.0386 ± 0.0022 & 0.0419 ± 0.0007 & 0.0652 ± 0.0006 & 0.0418 ± 0.0014 & \underline{0.0704 ± 0.0021} & 0.0608 ± 0.0016 & 0.0687 ± 0.0016 & 0.0474 ± 0.0010 \\
& NDCG@5 & 0.0082 ± 0.0006 & 0.0086 ± 0.0002 & 0.0152 ± 0.0002 & 0.0091 ± 0.0003 & 0.0159 ± 0.0003 & 0.0134 ± 0.0004 & 0.0162 ± 0.0003 & 0.0102 ± 0.0004 \\
& NDCG@10 & 0.0113 ± 0.0008 & 0.0120 ± 0.0001 & 0.0207 ± 0.0003 & 0.0125 ± 0.0003 & 0.0217 ± 0.0006 & 0.0181 ± 0.0006 & 0.0216 ± 0.0003 & 0.0140 ± 0.0007 \\
& NDCG@20 & 0.0153 ± 0.0009 & 0.0164 ± 0.0003 & 0.0267 ± 0.0004 & 0.0166 ± 0.0004 & 0.0285 ± 0.0007 & 0.0243 ± 0.0008 & 0.0282 ± 0.0005 & 0.0188 ± 0.0006 \\

\specialrule{1pt}{1pt}{1pt} 

    \end{tabular}%
    }
    \resizebox{\textwidth}{!}{%
    % \begin{tabular}{c|l|ccccccc|c|c}
    \begin{tabular}{ll| cccccccc}
    \specialrule{1pt}{1pt}{2pt}
    % \multicolumn{2}{c|}{Dataset} & \multicolumn{6}{c|}{\texttt{Electricity}}&\multicolumn{6}{c}{\texttt{Energy}} \\
    Dataset & Metric & CL4SRec & CoSeRec & ICLRec & DuoRec & ICSRec & MVS & FENRec & Improv.   \\
    \specialrule{1pt}{1pt}{2pt}

\multirow{6}{*}{Sports} 
 & HIT@5 & 0.0235 ± 0.0003 & 0.0264 ± 0.0002 & 0.0271 ± 0.0005 & 0.0311 ± 0.0003 & 0.0388 ± 0.0005 & 0.0384 ± 0.0006 & \textbf{0.0431 ± 0.0008}\textsuperscript{*} & 7.75\% \\
& HIT@10 & 0.0375 ± 0.0003 & 0.0403 ± 0.0003 & 0.0422 ± 0.0006 & 0.0446 ± 0.0006 & 0.0551 ± 0.0009 & 0.0548 ± 0.0002 & \textbf{0.0621 ± 0.0005}\textsuperscript{*} & 6.52\%\\
& HIT@20 & 0.0575 ± 0.0015 & 0.0605 ± 0.0009 & 0.0632 ± 0.0012 & 0.0640 ± 0.0002 & 0.0767 ± 0.0006 & 0.0775 ± 0.0002 & \textbf{0.0890 ± 0.0007}\textsuperscript{*} & 7.23\%\\
& NDCG@5 & 0.0156 ± 0.0002 & 0.0177 ± 0.0002 & 0.0179 ± 0.0006 & 0.0220 ± 0.0005 & 0.0272 ± 0.0006 & 0.0268 ± 0.0004 & \textbf{0.0299 ± 0.0008}\textsuperscript{*} & 6.79\% \\
& NDCG@10 & 0.0201 ± 0.0001 & 0.0221 ± 0.0001 & 0.0227 ± 0.0007 & 0.0263 ± 0.0002 & 0.0324 ± 0.0006 & 0.0321 ± 0.0003 & \textbf{0.0361 ± 0.0006}\textsuperscript{*} & 6.49\% \\
& NDCG@20 & 0.0251 ± 0.0003 & 0.0272 ± 0.0003 & 0.0280 ± 0.0008 & 0.0312 ± 0.0003 & 0.0379 ± 0.0005 & 0.0378 ± 0.0003 & \textbf{0.0429 ± 0.0007}\textsuperscript{*} & 6.98\% \\

\specialrule{0.5pt}{1pt}{2pt} 

\multirow{6}{*}{Beauty} 
& HIT@5 & 0.0492 ± 0.0006 & 0.0459 ± 0.0020 & 0.0493 ± 0.0016 & 0.0560 ± 0.0005 & 0.0681 ± 0.0010 & 0.0691 ± 0.0008 & \textbf{0.0728 ± 0.0008}\textsuperscript{*} & 2.97\% \\
& HIT@10 & 0.0706 ± 0.0016 & 0.0696 ± 0.0017 & 0.0726 ± 0.0015 & 0.0800 ± 0.0011 & 0.0936 ± 0.0013 & 0.0961 ± 0.0006 & \textbf{0.1019 ± 0.0007}\textsuperscript{*} & 4.19\% \\
& HIT@20 & 0.0990 ± 0.0013 & 0.1020 ± 0.0007 & 0.1055 ± 0.0029 & 0.1088 ± 0.0012 & 0.1273 ± 0.0002 & 0.1305 ± 0.0010 & \textbf{0.1393 ± 0.0030} & 3.57\% \\
& NDCG@5 & 0.0348 ± 0.0005 & 0.0301 ± 0.0013 & 0.0325 ± 0.0008 & 0.0406 ± 0.0004 & 0.0487 ± 0.0004 & 0.0494 ± 0.0002 & \textbf{0.0514 ± 0.0006}\textsuperscript{*} & 2.19\% \\
& NDCG@10 & 0.0417 ± 0.0008 & 0.0378 ± 0.0012 & 0.0400 ± 0.0008 & 0.0483 ± 0.0005 & 0.0569 ± 0.0008 & 0.0581 ± 0.0001 & \textbf{0.0608 ± 0.0005}\textsuperscript{*} & 3.05\% \\
& NDCG@20 & 0.0488 ± 0.0008 & 0.0460 ± 0.0009 & 0.0483 ± 0.0012 & 0.0555 ± 0.0006 & 0.0654 ± 0.0005 & 0.0667 ± 0.0002 & \textbf{0.0702 ± 0.0011} & 2.93\% \\
\specialrule{0.5pt}{1pt}{2pt} 

\multirow{6}{*}{Toys} 
 & HIT@5 & 0.0630 ± 0.0009 & 0.0576 ± 0.0008 & 0.0576 ± 0.0008 & 0.0609 ± 0.0007 & 0.0776 ± 0.0009 & 0.0748 ± 0.0010 & \textbf{0.0818 ± 0.0010} & 3.28\% \\
& HIT@10 & 0.0863 ± 0.0001 & 0.0818 ± 0.0014 & 0.0826 ± 0.0020 & 0.0816 ± 0.0008 & 0.1035 ± 0.0012 & 0.1008 ± 0.0013 & \textbf{0.1109 ± 0.0007}\textsuperscript{*} & 4.03\% \\
& HIT@20 & 0.1143 ± 0.0000 & 0.1121 ± 0.0006 & 0.1137 ± 0.0027 & 0.1080 ± 0.0008 & 0.1355 ± 0.0015 & 0.1323 ± 0.0030 & \textbf{0.1462 ± 0.0010}\textsuperscript{*} & 4.06\% \\
& NDCG@5 & 0.0447 ± 0.0004 & 0.0399 ± 0.0003 & 0.0393 ± 0.0005 & 0.0449 ± 0.0003 & 0.0566 ± 0.0003 & 0.0547 ± 0.0004 & \textbf{0.0592 ± 0.0004} & 3.14\% \\
& NDCG@10 & 0.0522 ± 0.0002 & 0.0477 ± 0.0004 & 0.0473 ± 0.0007 & 0.0515 ± 0.0003 & 0.0650 ± 0.0005 & 0.0631 ± 0.0005 & \textbf{0.0686 ± 0.0003}\textsuperscript{*} & 3.63\% \\
& NDCG@20 & 0.0592 ± 0.0002 & 0.0553 ± 0.0001 & 0.0552 ± 0.0012 & 0.0582 ± 0.0003 & 0.0731 ± 0.0005 & 0.0710 ± 0.0009 & \textbf{0.0775 ± 0.0002}\textsuperscript{*} & 3.75\% \\

\specialrule{0.5pt}{1pt}{2pt} 

\multirow{6}{*}{Yelp} 
& HIT@5 & 0.0238 ± 0.0005 & 0.0221 ± 0.0003 & 0.0232 ± 0.0005 & 0.0236 ± 0.0002 & \underline{0.0260 ± 0.0005} & 0.0243 ± 0.0004 & \textbf{0.0286 ± 0.0009}\textsuperscript{*} & 10.00\% \\
& HIT@10 & 0.0404 ± 0.0014 & 0.0375 ± 0.0008 & 0.0394 ± 0.0006 & 0.0402 ± 0.0011 & 0.0431 ± 0.0003 & 0.0409 ± 0.0007 & \textbf{0.0485 ± 0.0008}\textsuperscript{*} & 12.27\% \\
& HIT@20 & 0.0655 ± 0.0012 & 0.0618 ± 0.0004 & 0.0645 ± 0.0013 & 0.0663 ± 0.0018 & 0.0700 ± 0.0002 & 0.0654 ± 0.0001 & \textbf{0.0776 ± 0.0020}\textsuperscript{*} & 10.23\% \\
& NDCG@5 & 0.0150 ± 0.0003 & 0.0141 ± 0.0002 & 0.0147 ± 0.0004 & 0.0150 ± 0.0002 & \underline{0.0165 ± 0.0002} & 0.0156 ± 0.0003 & \textbf{0.0182 ± 0.0003}\textsuperscript{*} & 10.30\% \\
& NDCG@10 & 0.0204 ± 0.0006 & 0.0190 ± 0.0003 & 0.0198 ± 0.0001 & 0.0202 ± 0.0004 & \underline{0.0220 ± 0.0002} & 0.0210 ± 0.0004 & \textbf{0.0246 ± 0.0002}\textsuperscript{*} & 11.82\% \\
& NDCG@20 & 0.0266 ± 0.0006 & 0.0251 ± 0.0003 & 0.0262 ± 0.0003 & 0.0268 ± 0.0006 & \underline{0.0288 ± 0.0002} & 0.0271 ± 0.0002 & \textbf{0.0319 ± 0.0005}\textsuperscript{*} & 10.76\% \\

\specialrule{1pt}{1pt}{1pt} 

    \end{tabular}%
    }
    \caption{Performance comparison of different methods on 4 datasets with standard deviations. The best results are in boldface and the second-best results
    are underlined. `Improv.' indicates the relative improvement against the best baseline performance. `*' denotes the significance p-value $<$ 0.05 compared with the best baseline.}
    \label{tab:main_exp1}
\end{table*}

\subsection{Comparison to SOTA}

Table~\ref{tab:main_exp1} presents a performance comparison of different methods across four datasets, including standard deviations. \textit{In addition to the information provided in the main text, we also provide the standard deviations here.} Due to space constraints, the table has been placed at the end of this document. Experimental results demonstrate that our method, FENRec, achieves significant improvements with the p-value $<$ 0.05 across most metrics, highlighting its effectiveness.

\subsection{Hyperparameter Tuning}

% \begin{figure}[htbp]
%     \centering
%     \includegraphics[width=0.5\columnwidth]{images/lambda.png}
%     \caption{Performance of FENRec w.r.t. different $\lambda$ values}
%     \label{fig:tuning_lambda}
% \end{figure}

\begin{figure}[htbp]
    \centering
    \begin{minipage}[b]{0.45\columnwidth}
        \centering
        \includegraphics[width=\linewidth]{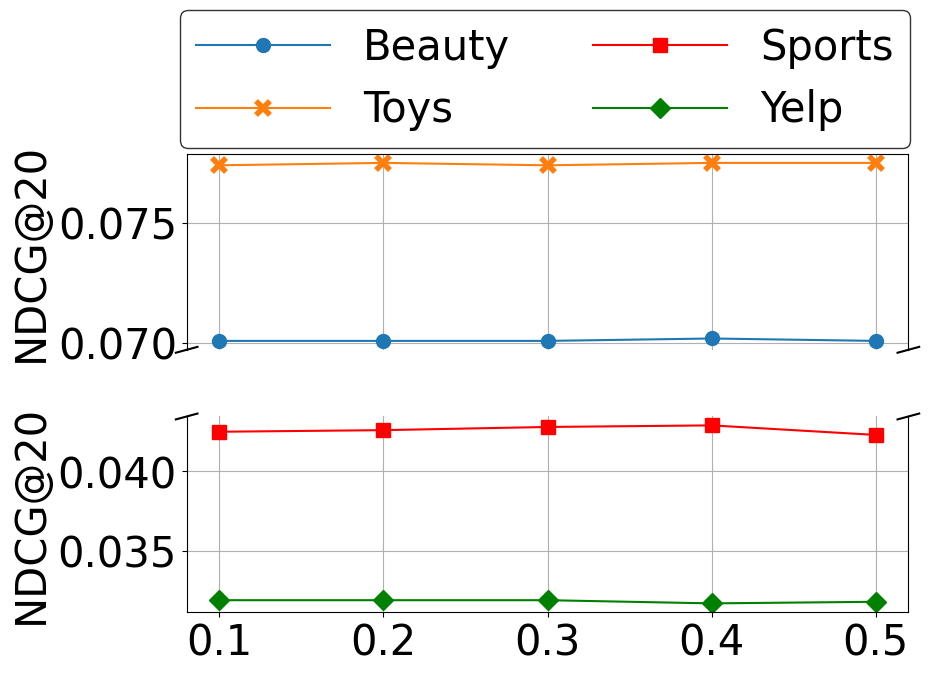}
        \subcaption{$\gamma$}\label{fig:tuning_lambda}
    \end{minipage}
    \hfill
    \begin{minipage}[b]{0.45\columnwidth}
        \centering
        \includegraphics[width=\linewidth]{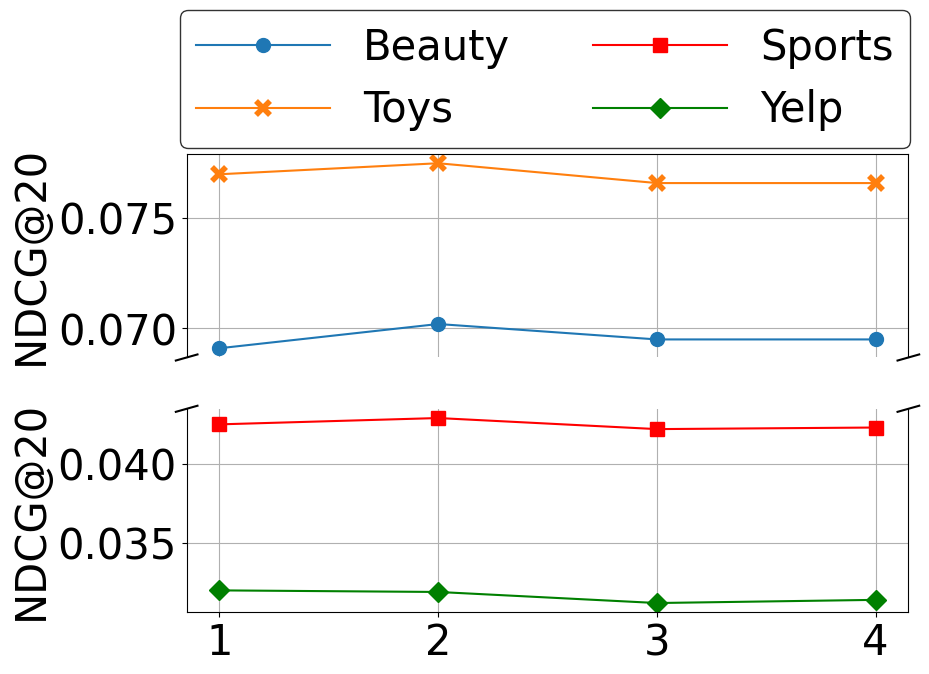}
        \subcaption{Future items}\label{fig:tuning_label_len}
    \end{minipage}
    \caption{(a): Performance of FENRec w.r.t. different $\lambda$ values. (b): Performance of FENRec w.r.t. the number of future items considered beyond the immediate next item.}
    \label{fig:tuning_app}
\end{figure}

Figure~\ref{fig:tuning_lambda} shows the results of tuning $\lambda$, a hyperparameter between 0 and 1, that controls the proportion of the anchor sample in the generated enduring hard negatives. The results show that the improvement of our method, FENRec, is relatively insensitive to changes in $\lambda$. This is likely due to the incorporation of time-dependent soft labeling, which mitigates the potential risk of misclassifying genuinely appealing items as irrelevant, a risk that may increase with higher $\lambda$ values. By effectively capturing and representing the ambiguity in user interest, time-dependent soft labeling ensures that the model maintains stable performance as $\lambda$ increases. The combined effects of enduring hard negatives and time-dependent soft labeling support the robustness and effectiveness of our method across different $\lambda$ settings. 

To effectively leverage future interactions without introducing excessive noise, we designed time-dependent soft labels that incorporate two future items beyond the immediate next item. This approach ensures a balance between capturing meaningful user intent while minimizing the influence of less relevant, distant interactions.
To validate this design, we conducted hyperparameter tuning on the Sports, Beauty, Toys, and Yelp datasets, varying the number of future items included. Figure~\ref{fig:tuning_label_len} illustrates the performance variations, where the x-axis represents the number of future items considered beyond the immediate next item (i.e., excluding the immediate next item), and the y-axis denotes NDCG@20. As shown in Figure~\ref{fig:tuning_label_len}, the model achieved the best performance with two future items on the Beauty, Sports, and Toys datasets, while performing second-best on the Yelp dataset. The results demonstrate that this configuration effectively captures relevant future interactions without introducing unnecessary noise.

\clearpage
\clearpage

\end{appendices}

\end{document}